\pdfoutput=1
\documentclass[11pt]{article}
\usepackage[margin=1in]{geometry}
\usepackage[T1]{fontenc}
\usepackage{lmodern}
\usepackage{amsmath,amssymb,amsthm,mathtools}
\usepackage{booktabs,tabularx,array}
\newcolumntype{Y}{>{\raggedright\arraybackslash}X}
\usepackage{enumitem}
\usepackage[dvipsnames]{xcolor}
\usepackage[colorlinks=true,linkcolor=blue,citecolor=blue,urlcolor=blue]{hyperref}
\usepackage[capitalise,nameinlink]{cleveref}
\usepackage{natbib}
\usepackage{tikz}
\usetikzlibrary{positioning,arrows.meta,calc,fit,backgrounds,shapes.geometric}
\usepackage{pgfplots}
\pgfplotsset{compat=1.18}
\usepgfplotslibrary{groupplots}
\definecolor{cBlue}{HTML}{2A78D6}
\definecolor{cOrange}{HTML}{EB6834}
\definecolor{cAqua}{HTML}{1BAF7A}
\definecolor{cInk}{HTML}{52514E}
\definecolor{cGrid}{HTML}{D9D8D3}
\definecolor{cAdm}{HTML}{E6F2EC}
\pgfplotsset{paper/.style={
  axis lines=left, axis line style={cInk}, tick style={cInk}, tick label style={font=\footnotesize,color=cInk},
  label style={font=\small}, grid=major, grid style={cGrid,very thin}, legend style={font=\footnotesize,draw=none,fill=none},
  legend cell align=left}}

\theoremstyle{plain}
\newtheorem{theorem}{Theorem}[section]
\newtheorem{proposition}[theorem]{Proposition}
\newtheorem{corollary}[theorem]{Corollary}

\theoremstyle{definition}
\newtheorem{definition}[theorem]{Definition}
\newtheorem{problem}[theorem]{Problem}
\newtheorem{example}[theorem]{Example}
\theoremstyle{remark}
\newtheorem{remark}[theorem]{Remark}

\newcommand{\D}{D}
\newcommand{\Qstar}{Q^{*}}
\newcommand{\QJ}{Q^{J}}
\newcommand{\sstar}{s^{*}}
\newcommand{\shat}{\hat{s}}
\newcommand{\Pol}{\Pi}
\newcommand{\Adm}{\mathcal{A}}
\newcommand{\PolClass}{\mathcal{C}}
\newcommand{\Dist}{\mathcal{D}}
\newcommand{\E}{\mathbb{E}}
\newcommand{\Rec}{\mathrm{Rec}}
\newcommand{\Prec}{\mathrm{Prec}}
\newcommand{\TP}{\mathrm{TP}}
\newcommand{\FP}{\mathrm{FP}}
\newcommand{\FN}{\mathrm{FN}}
\newcommand{\Cost}{C}
\newcommand{\oeq}{\equiv^{*}}
\newcommand{\semsel}{\sigma_{s}}
\newcommand{\semjoin}{\bowtie_{s}}
\newcommand{\act}[1]{\textup{\textsc{#1}}}

\title{When Plans Change Answers:\\Formalizing Cost--Accuracy Optimization for Semantic Queries}
\author{Kyoungmin Kim\\KAIST}
\date{}

\begin{document}
\maketitle

\begin{abstract}
In semantic query engines, predicates are evaluated by machine-learned models, and the choice of a query plan affects not only the cost of a query but also its result. Existing systems either apply a fixed threshold to each semantic operator or tune accuracy per operator, without accounting for how errors propagate through joins. We give a formal problem definition for cost--accuracy optimization of such queries. Our starting point is the calibrated confidence that decision models such as Jev attach to each decision. It yields an expected error for every decision; weighting these errors by each decision's contribution to the output (in the simplest case, its fan-out) gives the expected output quality of a plan without any labeled data, and the same computation in reverse turns an output-level accuracy target into a price on each base or intermediate tuple. Building on this, we define an oracle semantics for relational algebra with semantic operators, physical plans as pairs of a logical plan and a decision policy, declarative output-level targets, and a hierarchy of plan equivalence. We show that accuracy is plan-invariant under pointwise-deterministic policies, and that selection pushdown is not quality-sound when escalation bands are calibrated on the plan's own candidates. Expected quality can be computed in polynomial time under bag semantics; under set semantics it follows the dichotomy of tuple-independent probabilistic databases when every relation carries a semantic predicate. Choosing which tuples to drop is NP-hard, while the optimization problem decomposes into per-tuple decisions through two Lagrange multipliers, and, with what we call confidence-centric skipping, tuples that can no longer affect the target are skipped without being scored. Simulations on a synthetic workload illustrate these effects; an evaluation on real engines is left for future work.
\end{abstract}

\section{Introduction}\label{sec:intro}

A classical query optimizer can treat all plans in its search space as interchangeable with respect to the result and compare them by cost alone. Semantic query engines break this assumption. Once a predicate is evaluated by a model, thresholds, cascades, candidate pruning, and sample-based calibration all affect which tuples are returned, and two plans for the same query may return different answers. In this paper we ask three questions. When should two such plans be considered equivalent? How can we predict the accuracy of a plan without comparing its output to ground truth? And how do we find the cheapest plan whose accuracy is good enough?

\paragraph{Running example.}
Consider a retailer that wants to contact every customer who bought a product for which some review reports a battery fire:
\begin{verbatim}
  SELECT o.*  FROM Reviews r JOIN Orders o ON r.product = o.product
  WHERE  SEM(r.text, 'reports a battery fire')
\end{verbatim}
We assume that the semantic predicate is evaluated by an inexpensive decision model, such as Jev in JEVDB~\citep{jevdb2026}, which returns a confidence $p_r$ that review $r$ satisfies the predicate. Uncertain reviews can additionally be sent to a more expensive LLM judge. A review that passes the filter contributes one output row per order of its product; we call this number its \emph{fan-out} $f_r$. \Cref{tab:toy} lists five reviews.

\begin{table}[t]
\centering
\caption{Running example: five reviews, their decision-model confidence $p$, and fan-out $f$ (orders of the reviewed product). A plan that accepts the reviews with $p\ge0.5$ makes the expected errors on the right: per review (operator level) and per output row (output level).}\label{tab:toy}
\small
\begin{tabular}{@{}lrrcll@{}}
\toprule
Review & $p$ & $f$ & Action & Expected error per review & Expected error in output rows \\
\midrule
$r_1$ & 0.98 & 2   & \act{Accept} & FP: $1-p=0.02$ & FP: $0.02\cdot2=0.04$ \\
$r_2$ & 0.90 & 40  & \act{Accept} & FP: $0.10$ & FP: $0.10\cdot40=4.00$ \\
$r_3$ & 0.60 & 5   & \act{Accept} & FP: $0.40$ & FP: $0.40\cdot5=2.00$ \\
$r_4$ & 0.30 & 100 & \act{Reject} & FN: $p=0.30$ & FN: $0.30\cdot100=\mathbf{30.00}$ \\
$r_5$ & 0.05 & 3   & \act{Reject} & FN: $0.05$ & FN: $0.05\cdot3=0.15$ \\
\midrule
\multicolumn{4}{@{}l}{Expected recall\,/\,precision} & $0.88$\,/\,$0.83$ & $\mathbf{0.58}$\,/\,$0.87$ \\
\bottomrule
\end{tabular}
\end{table}

\paragraph{From confidence to expected error.}
Suppose the confidences are calibrated. Then a review accepted with confidence $p$ is a false positive with probability $1-p$, and a rejected review is a false negative with probability $p$. Let $\mathrm{conf}_i$ denote the confidence in the action that was taken, i.e., $p_i$ if the review is accepted and $1-p_i$ if it is rejected. By linearity of expectation,
\[
\E[\#\text{wrong decisions}]=\sum_i\big(1-\mathrm{conf}_i\big),
\]
where accepted instances account for the false positives and rejected instances for the false negatives. For the plan in \cref{tab:toy}, the filter makes $0.52$ false positives and $0.35$ false negatives in expectation, for an expected recall of $2.48/(2.48+0.35)=0.88$. The query result, however, consists of orders and not of reviews. An error on $r_i$ affects $f_i$ output rows, so the expected number of wrong output rows is the same sum weighted by each decision's contribution to the output, here $\sum_i f_i(1-\mathrm{conf}_i)$. On output rows the expected recall drops to $0.58$, and $30$ of the $30.15$ expected missing rows are due to a single review, $r_4$, which is both uncertain and has a large fan-out.

\paragraph{From expected error back to decisions.}
The same reasoning can be applied in the opposite direction to make decisions. Escalating $r_4$ to the judge costs one LLM call and removes $30$ expected missing rows, which raises the expected output recall from $0.58$ to $0.998$. Escalating $r_3$ also costs one call, but removes only $2$ expected wrong rows. In general, an output-level target can be \emph{pushed down} to a price per unit of expected output error. Each base or intermediate tuple then compares the cost of resolving its uncertainty with the reduction in output error that this buys. \Cref{fig:chain} illustrates both directions.

\begin{figure}[t]
\centering
\begin{tikzpicture}[
  box/.style={draw,rounded corners,align=center,inner sep=4pt,font=\small,minimum height=1.25cm},
  arr/.style={-{Stealth[length=2.2mm]},thick},
  lab/.style={font=\footnotesize,align=center}]
\node[box] (c) {calibrated\\confidence $p_i$\\per decision};
\node[box,right=1.3cm of c] (e) {expected decision error\\$1-\mathrm{conf}_i$\\(FP if accepted, FN if rejected)};
\node[box,right=1.3cm of e] (o) {expected output error\\$\sum_i c_i\,(1-\mathrm{conf}_i)$\\$\Rightarrow\ \E[\Rec],\ \E[\Prec]$};
\draw[arr] (c) -- (e);
\draw[arr] (e) -- node[lab,above]{$\times\,c_i$} node[lab,below]{lineage} (o);
\draw[arr,dashed] (o.south) to[out=-150,in=-30] node[lab,below]{target $(\rho,\gamma)$ pushed down as a price per unit of output error,\\weighed per tuple against the cost of an escalation, a drop, or a reorder} (c.south);
\end{tikzpicture}
\caption{The role of calibrated confidence. Forward (solid): each decision has an expected error, which lineage weights by the decision's contribution $c_i$ to the output (in the running example, its fan-out); the sum gives the expected output quality of a plan. Backward (dashed): an output-level target induces a cost--accuracy trade-off for each tuple.}\label{fig:chain}
\end{figure}
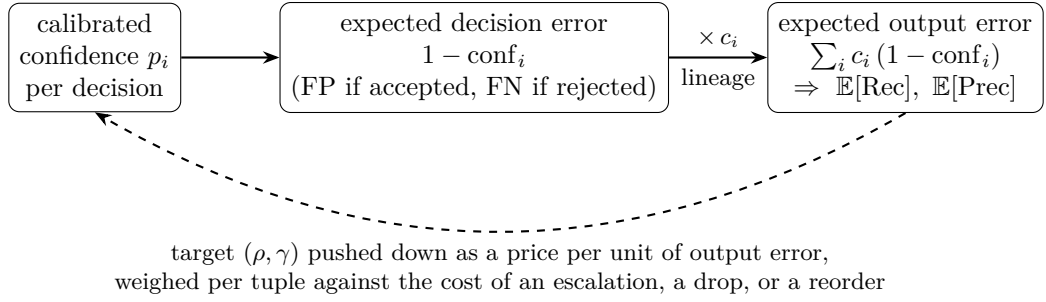

\paragraph{Why calibrated confidence matters.}
Without confidences, the accuracy of a plan can only be measured against labeled ground truth, and an optimizer has nothing but cost to compare plans by. Boolean decisions from a fixed model and threshold do not help either, since then all classically equivalent plans return the same answer (\cref{prop:invariance}). Calibrated confidences allow us to predict the output quality of a plan before executing it (\cref{sec:quality,sec:estimation}), to compare plans by their predicted quality and thereby decide which rewrites are safe (\cref{sec:equivalence}), and to optimize by assigning a price to output error at the level of individual tuples (\cref{sec:propagation,sec:optimization}).

\paragraph{Motivating gaps.}
JEVDB~\citep{jevdb2026} evaluates predicates with Jev either at a fixed threshold (its Flash mode) or with a per-predicate escalation band $[\ell,h)$. The band is fit on a sample of about $200$ candidates labeled by an LLM judge, such that the expected recall and precision loss relative to the judge stay below user-given tolerances ($0.05$ in the experiments). JEVDB pushes semantic filters down by default, arguing that Jev calls are cheap, and leaves semantic-aware query optimization to future work. Like other current systems, it states accuracy targets per operator.

Per-operator targets can differ substantially from the quality of the query result. In \cref{tab:toy}, the filter has an expected recall of $0.88$, whereas the query result has an expected recall of $0.58$. To our knowledge, no current system lets the user specify a single accuracy target for the query and distributes it across operators and tuples.

\paragraph{Contributions.}
\begin{enumerate}[leftmargin=*]
  \item A formal model: oracle semantics, physical plans as (logical plan, decision policy), and output-level quality targets, with a plan's expected output quality derived from calibrated confidences (\cref{sec:model,sec:plans,sec:quality}).
  \item A hierarchy of semantic plan equivalence: a plan-invariance result for pointwise-deterministic implementations, a confidence-relative equivalence, and quality-sound rewrites with explicit counterexamples (\cref{sec:equivalence}).
  \item A lineage-based account of how output error propagates to intermediate and base tuples, and how an output target is pushed down to a per-tuple price (\cref{sec:propagation}).
  \item Estimators of plan quality with judge-relative and oracle-relative guarantees, and a dichotomy for exact quality evaluation (\cref{sec:estimation}).
  \item The cost--accuracy optimization problem, its per-tuple Lagrangian decomposition, confidence-centric skipping of instances that can no longer affect the target, without scoring them, and complexity results (\cref{sec:optimization}).
  \item Simulations on a synthetic workload with stated parameters that illustrate plan invariance, plan dependence under calibrated bands, and the accuracy of confidence-based quality prediction. In this workload, plans for the same query differ by three orders of magnitude in cost and cover the whole range of output quality, and per-tuple pricing meets an output target at several times lower cost than an operator-level band (\cref{sec:experiments}, \cref{fig:planspace,fig:sweep}). The simulations illustrate the mechanisms; they are not measurements of a deployed system.
\end{enumerate}

\section{Data Model and Semantic Relational Algebra}\label{sec:model}

We define the correct answer of a query by evaluating it with oracle semantic predicates. Errors can then only come from the evaluation of semantic predicates, never from relational operators.

\begin{definition}[Semantic operators]\label{def:semop}
A \emph{semantic predicate} is a function $s : \mathrm{dom}(x_1)\times\cdots\times\mathrm{dom}(x_k)\to\{0,1\}$ over text-valued (or multimodal) attributes, specified by a natural-language instruction $\phi$. A \emph{semantic labeling function} $c:\mathrm{dom}(x)\to C$ maps to a finite label set $C$; a \emph{semantic score} maps to an ordered set $L$. We write $\sstar$ for the oracle (intended) interpretation of $\phi$.
\end{definition}

\begin{definition}[Semantic relational algebra, SRA]\label{def:sra}
SRA extends positive relational algebra $(\sigma,\pi,\bowtie,\cup)$ with semantic selection $\semsel$, semantic join $\semjoin$, semantic projection $\pi_c$ (classification), and semantic ranking $\tau_k$. Aggregation and negation are deferred (\cref{sec:open}).
\end{definition}

\begin{definition}[Oracle semantics]\label{def:oracle}
For a query $Q$ in SRA and database $\D$, $\Qstar(\D)$ denotes the result of evaluating $Q$ with every semantic operator interpreted by its oracle. All plans are evaluated against $\Qstar(\D)$.
\end{definition}

\begin{example}[Oracle semantics in the running example]\label{ex:oracle}
The query of \cref{sec:intro} is
\[
Q=\semsel(\textit{Reviews})\bowtie_{\textit{product}}\textit{Orders},\qquad \phi=\text{``reports a battery fire''}.
\]
Suppose that, among the reviews in \cref{tab:toy}, exactly $r_1$, $r_2$, and $r_4$ satisfy $\phi$. Then $\Qstar(\D)$ is the bag of the $2+40+100=142$ orders of those products. A plan that accepts $r_1,r_2,r_3$ returns $2+40+5=47$ orders, of which $42$ are in $\Qstar(\D)$. The confidences $p_i$ in \cref{tab:toy} express the model's belief about these unknown labels, and \cref{sec:quality} uses them to compute expected quality.
\end{example}

\begin{remark}\label{rem:classical}
Because every $\sstar$ is a deterministic function of its inputs, every classical equivalence of positive relational algebra remains valid under oracle semantics (selection pushdown, join commutativity and associativity, semijoin reduction). Hence, the logical search space is the same as in the classical setting, and approximation only enters at the physical level.
\end{remark}

\paragraph{Bag or set outputs.}
We use bag semantics by default because it makes fan-out explicit; in \cref{ex:oracle}, every order is a separate answer. Set semantics (e.g., \texttt{SELECT DISTINCT o.customer}) is a variant in which an answer survives if any of its derivations survives (\cref{sec:propagation,sec:estimation}).

\section{Physical Plans and Decision Policies}\label{sec:plans}

We model a physical plan as a logical plan together with a decision policy, which determines for every semantic decision whether it is accepted, rejected, escalated, or skipped.

\begin{definition}[Decision instance]\label{def:instance}
For a semantic operator $s$ in plan $L$ over database $\D$, a \emph{decision instance} is a pair $(s,\kappa)$ where $\kappa$ is the tuple of attribute values $s$ reads (for a join, the pair of texts). Let $I_L(\D)$ be the multiset of decision instances that $L$ presents to semantic operators.
\end{definition}

\begin{example}[Instances depend on the plan]\label{ex:instances}
In the running example, the pushdown plan $L_{\downarrow}=\semsel(\textit{Reviews})\bowtie\textit{Orders}$ presents five instances, one per review. The pull-up plan $L_{\uparrow}=\semsel(\textit{Reviews}\bowtie\textit{Orders})$ evaluates the predicate on each joined row, so it presents $2+40+5+100+3=150$ instances; review $r_4$'s text appears $100$ times. The two logical plans are equivalent, but $I_{L_\downarrow}(\D)$ and $I_{L_\uparrow}(\D)$ are different multisets. This difference becomes relevant for policies that depend on the population of instances (\cref{tab:policy-classes}).
\end{example}

\begin{definition}[Decision model]\label{def:model}
A decision model $M$ returns, for an instance $(s,\kappa)$, a score $p_M(s,\kappa)\in[0,1]$. $M$ is \emph{calibrated} w.r.t.\ a distribution $\Dist$ over instances if $\Pr_{\kappa\sim\Dist}[\sstar(\kappa)=1 \mid p_M=p]=p$. Examples include Jev's \texttt{Noul} primitive, cross-encoders, and LLM judges that report verbalized or logit-based confidences.
\end{definition}

Intuitively, calibration means that about $30\%$ of all reviews with $p=0.3$ actually report a fire. Rejecting $r_4$ in \cref{tab:toy} is therefore wrong with probability $0.3$.

\begin{definition}[Decision policy]\label{def:policy}
A policy $\Pol$ maps each decision instance, together with a context, to one of four actions:
\begin{itemize}[leftmargin=*,nosep]
  \item \act{Accept}: treat the instance as satisfying $s$ (keep the tuple);
  \item \act{Reject}: treat it as not satisfying $s$ (drop the tuple);
  \item \act{Escalate}$(M_k)$: ask a more expensive model $M_k$ (e.g., the LLM judge) and follow its answer;
  \item \act{Skip}: drop the tuple without evaluating it, e.g., after a negative probe of a semantic Bloom filter, or as a cost-motivated drop.
\end{itemize}
The context may include the instance's score, its downstream fan-out, its position in the plan, and statistics observed so far.
\end{definition}

\begin{example}[Policies in the running example]\label{ex:policy}
The policy of \cref{tab:toy} is ``\act{Accept} if $p\ge0.5$, else \act{Reject}''. A JEVDB-style band policy with band $[0.2,0.8)$ accepts $r_1,r_2$, rejects $r_5$, and escalates $r_3$ and $r_4$ to the LLM. A context-aware policy might escalate only $r_4$, whose large fan-out makes an error costly, and accept $r_3$ even though $r_3$ is more likely to be wrong.
\end{example}

\begin{definition}[Physical plan]\label{def:plan}
$P=(L,\Pol)$. Its output $P(\D)$ is a random variable when $\Pol$ uses sampling (e.g., threshold calibration) or a nondeterministic model.
\end{definition}

\Cref{tab:policy-classes} lists the policy classes used to state results.

\begin{table}[t]
\centering
\caption{Policy classes.}\label{tab:policy-classes}
\begin{tabularx}{\linewidth}{@{}lYY@{}}
\toprule
Class & Action depends on & Example \\
\midrule
Pointwise-deterministic (PD) & $\kappa$ only, fixed function & Jev with fixed $\tau$; JEVDB-Flash; ``accept iff $p\ge0.5$'' \\
Population-calibrated (PC) & $\kappa$ and a statistic of $I_L(\D)$ & escalation band $[\ell,h)$ fit on a sample of current candidates; JEVDB's quality-controlled mode \\
Context-aware (CA) & $\kappa$, fan-out, downstream cost/selectivity & escalate $r_4$ because $f_4=100$ \\
Sound pruning & necessary condition, sound extractor & lossless SBF; relational Bloom semijoin \\
\bottomrule
\end{tabularx}
\end{table}

\begin{definition}[Cost]\label{def:cost}
$\Cost(P,\D)=\sum_{\text{model invocations}}\mathrm{price}(M)\cdot\mathrm{tokens}+\Cost_{\mathrm{rel}}(P,\D)$. Relational operators do not introduce errors, but they are not free. Each accepted review adds its $f$ orders to the join and to any downstream work, such as composing one notification per order. Latency can be used instead of, or in addition to, monetary cost.
\end{definition}

\begin{example}[Cost in the running example]\label{ex:cost}
We charge one unit per decision-model call, $50$ per LLM call, and $0.05$ per joined row (the values used in \cref{sec:experiments}). With the band policy of \cref{ex:policy}, the pushdown plan costs $5$ model calls, $2$ LLM calls, and the rows of the reviews it returns. The pull-up plan with the same band makes $150$ model calls and escalates every one of the $105$ rows of $r_3$ and $r_4$ unless decisions are cached by $\kappa$. This amounts to $150+105\cdot50=5{,}400$ units of model cost instead of $5+2\cdot50=105$, for the same result.
\end{example}

\section{Quality Semantics and Declarative Targets}\label{sec:quality}

We measure quality on the query result, relative to $\Qstar(\D)$. Accuracy of individual operators is an internal quantity that the optimizer may budget, but it is not what the user is promised. Calibrated confidences allow us to predict output quality without ground truth.

\begin{definition}[Output quality]\label{def:quality}
For bag outputs $O=P(\D)$ and $O^*=\Qstar(\D)$, with multiplicities and multiset intersection $\cap$,
\[
\Rec(O,O^*)=\frac{|O\cap O^*|}{|O^*|},\qquad \Prec(O,O^*)=\frac{|O\cap O^*|}{|O|}.
\]
Set versions apply after duplicate elimination. Ranking (top-$k$) and classification outputs use rank or label metrics; aggregates use relative error (deferred).
\end{definition}

In \cref{ex:oracle}, the plan returns $47$ orders, $42$ of which are correct, while there are $142$ correct orders in total. Hence $\Rec=42/142=0.30$ and $\Prec=42/47=0.89$. These values refer to one particular set of true labels; the confidences describe a distribution over such label sets.

\paragraph{Expected quality from calibrated confidences.}
Fix a plan $P$ and assume that every semantic decision $i$ comes with a calibrated confidence $p_i$. We define the confidence in the action taken as
\[
\mathrm{conf}_i=\begin{cases}p_i & \text{if } i \text{ is accepted},\\ 1-p_i & \text{if } i \text{ is rejected},\\ \alpha_J & \text{if } i \text{ is escalated to a judge of accuracy } \alpha_J,\end{cases}
\]
so that decision $i$ is wrong with probability $1-\mathrm{conf}_i$. Let $c_i$ denote the number of output tuples whose derivation uses decision $i$. In the running example, $c_i$ is the fan-out; \cref{def:contribution} gives the general definition.

\begin{definition}[Predicted quality]\label{def:predicted}
For a plan with one semantic operator under bag semantics, a perfect judge ($\alpha_J=1$), and accepted, rejected, and escalated instance sets $A,R,E$,
\[
\E[\TP]=\sum_{i\in A\cup E}c_i\,p_i,\qquad \E[\FP]=\sum_{i\in A}c_i\,(1-p_i),\qquad \E[\FN]=\sum_{i\in R}c_i\,p_i,
\]
so that $\E[\FP]+\E[\FN]=\sum_i c_i(1-\mathrm{conf}_i)$. The \emph{predicted recall and precision} are
\[
\widehat{\Rec}_M(P)=\frac{\E[\TP]}{\E[\TP]+\E[\FN]},\qquad \widehat{\Prec}_M(P)=\frac{\E[\TP]}{\E[\TP]+\E[\FP]} .
\]
Setting every $c_i=1$ gives the operator-level quantities. \Cref{sec:propagation} extends $c_i$ to several semantic operators and intermediate tuples, and \cref{sec:estimation} handles imperfect judges, set semantics, and miscalibration.
\end{definition}

The ratio of expectations is an approximation of $\E[\Rec]$ that becomes accurate when many decisions contribute to the output. We quantify the gap empirically in \cref{sec:experiments}.

\begin{example}[Operator-level vs.\ output-level prediction]\label{ex:predicted}
For the policy of \cref{tab:toy} and $c_i=1$, we obtain $\E[\TP]=0.98+0.90+0.60=2.48$, $\E[\FP]=0.52$, $\E[\FN]=0.30+0.05=0.35$, and hence $\widehat{\Rec}=0.88$ and $\widehat{\Prec}=0.83$. With $c_i=f_i$, we obtain $\E[\TP]=1.96+36+3=40.96$, $\E[\FP]=6.04$, $\E[\FN]=30+0.15=30.15$, and hence $\widehat{\Rec}=0.58$ and $\widehat{\Prec}=0.87$. If $r_4$ is escalated, it moves from $R$ to $E$; then $\E[\TP]=70.96$, $\E[\FN]=0.15$, $\widehat{\Rec}=0.998$, and $\widehat{\Prec}=0.92$. None of these numbers requires labeled data, although each requires a confidence for every review; \cref{sec:skip} shows when some reviews need not be scored at all. \Cref{fig:errors} breaks the errors down by review.
\end{example}

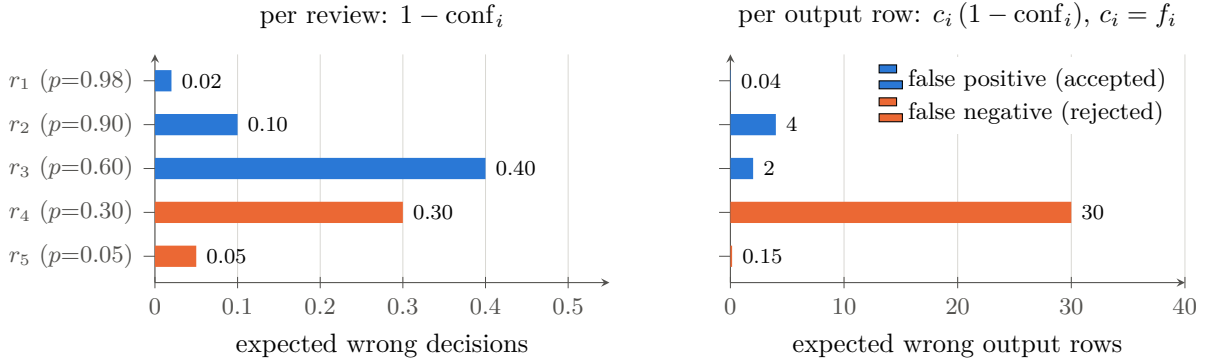
\begin{figure}[t]
\centering
\begin{tikzpicture}
\begin{groupplot}[
  group style={group size=2 by 1, horizontal sep=1.6cm, yticklabels at=edge left},
  paper, xbar, /pgf/bar width=8pt, /pgf/bar shift=0pt, width=0.46\linewidth, height=4.6cm,
  symbolic y coords={r5,r4,r3,r2,r1}, ytick={r5,r4,r3,r2,r1},
  yticklabels={$r_5$ ($p{=}0.05$),$r_4$ ($p{=}0.30$),$r_3$ ($p{=}0.60$),$r_2$ ($p{=}0.90$),$r_1$ ($p{=}0.98$)},
  xmin=0, enlarge y limits=0.15, ymajorgrids=false,
  nodes near coords, nodes near coords style={font=\scriptsize,color=black,anchor=west},
  every axis plot/.append style={draw=none}]
\nextgroupplot[xmax=0.55, title={\small per review: $1-\mathrm{conf}_i$},
  xlabel={expected wrong decisions}, nodes near coords={\pgfmathprintnumber[fixed,zerofill,precision=2]{\pgfplotspointmeta}}]
  \addplot[fill=cBlue] coordinates {(0.02,r1) (0.10,r2) (0.40,r3)};
  \addplot[fill=cOrange] coordinates {(0.30,r4) (0.05,r5)};
\nextgroupplot[xmax=40, title={\small per output row: $c_i\,(1-\mathrm{conf}_i)$, $c_i=f_i$},
  xlabel={expected wrong output rows}, nodes near coords={\pgfmathprintnumber[fixed,precision=2]{\pgfplotspointmeta}},
  legend style={at={(0.99,0.99)},anchor=north east,fill=white}]
  \addplot[fill=cBlue] coordinates {(0.04,r1) (4.0,r2) (2.0,r3)};
  \addplot[fill=cOrange] coordinates {(30.0,r4) (0.15,r5)};
  \legend{false positive (accepted), false negative (rejected)}
\end{groupplot}
\end{tikzpicture}
\caption{Expected errors of the plan in \cref{tab:toy}. Left: per review, $r_3$ (accepted with confidence $0.6$) has the highest error probability. Right: weighted by contribution (fan-out), $r_4$ dominates; its error probability of $0.3$ applies to $100$ orders and accounts for $30$ of the $30.15$ expected missing output rows. An operator-level target corresponds to the left panel, while output quality corresponds to the right panel.}\label{fig:errors}
\end{figure}

\begin{definition}[Declarative target]\label{def:target}
A target is $T=(\rho,\gamma,\delta)$. Plan $P$ satisfies $T$ on $\D$ if
\[
\Pr_{\Pol}\!\big[\Rec(P(\D),\Qstar(\D))\ge\rho \;\wedge\; \Prec(P(\D),\Qstar(\D))\ge\gamma\big]\ge 1-\delta .
\]
The \emph{expected} variant replaces the probabilistic constraint with $\E[\Rec]\ge\rho$ and $\E[\Prec]\ge\gamma$; the \emph{predicted} variant uses $\widehat{\Rec}_M\ge\rho$ and $\widehat{\Prec}_M\ge\gamma$. Only the predicted variant can be checked by an optimizer.
\end{definition}

For example, the plan in \cref{tab:toy} violates the predicted target $T=(0.9,0.9,\cdot)$ since $\widehat{\Rec}=0.58$, whereas the same plan with $r_4$ escalated satisfies it ($0.998$ and $0.92$).

\begin{definition}[Reference-relative quality]\label{def:judge}
In practice, $\Qstar(\D)$ cannot be observed. For a reference judge $J$ (e.g., a frontier LLM), let $\QJ(\D)$ be $Q$ evaluated with $J$ as the oracle. A guarantee is \emph{judge-relative} if stated against $\QJ$ and \emph{oracle-relative} if stated against $\Qstar$.
\end{definition}

Whether a judge-relative guarantee carries over to the oracle depends on which instances the judge gets wrong, weighted by their contribution. Let $\Delta$ be the set of decision instances on which $J$ and $\sstar$ disagree, let $c^*(i)$ be the number of tuples of $\Qstar(\D)$ whose derivation uses $i$, and $c^J(i)$ the same count for $\QJ(\D)$. Define the \emph{contribution-weighted judge error}
\[
\varepsilon_J=\frac{\sum_{i\in\Delta}c^*(i)+\sum_{i\in\Delta}c^J(i)}{|\Qstar(\D)|}.
\]

\begin{proposition}[Judge-to-oracle transfer]\label{prop:judge}
Under bag semantics, for any plan output $O$,
\begin{align*}
\Rec(O,\QJ(\D))\ge\rho\;&\Rightarrow\;\Rec(O,\Qstar(\D))\ge\rho-\varepsilon_J,\\
\Prec(O,\QJ(\D))\ge\gamma\;&\Rightarrow\;\Prec(O,\Qstar(\D))\ge\gamma-\frac{\sum_{i\in\Delta}c^J(i)}{|O|}.
\end{align*}
\end{proposition}
\begin{proof}
Write $O^*=\Qstar(\D)$ and $O^J=\QJ(\D)$. An output tuple whose derivation uses no instance of $\Delta$ is derived identically under $J$ and $\sstar$, so $|O^*\setminus O^J|\le\sum_{i\in\Delta}c^*(i)$ and $|O^J\setminus O^*|\le\sum_{i\in\Delta}c^J(i)$ (bag difference). For multisets, $|O\cap O^*|\ge|O\cap O^J|-|O^J\setminus O^*|$, and $|O\cap O^J|\ge\rho\,|O^J|\ge\rho\,(|O^*|-|O^*\setminus O^J|)$. Dividing by $|O^*|$ and using $\rho\le1$ gives $\Rec(O,O^*)\ge\rho-\varepsilon_J$. For precision, divide $|O\cap O^*|\ge\gamma|O|-|O^J\setminus O^*|$ by $|O|$.
\end{proof}

\begin{example}[Fan-out makes unweighted judge error meaningless]\label{ex:judge}
In \cref{ex:oracle}, suppose the judge disagrees with the oracle only on $r_4$, wrongly answering ``no''. This is one error among five reviews, i.e., $20\%$ without weighting. However, $c^*(r_4)=100$ of the $142$ true orders and $c^J(r_4)=0$, so $\varepsilon_J=100/142=0.70$, and a judge-relative recall of $0.95$ only guarantees an oracle recall of $0.25$. The bound is almost tight: the plan that returns exactly $\QJ(\D)$ has judge-relative recall $1$ and oracle recall $42/142=0.30$. Judge accuracy should therefore be measured with contribution weights, in particular on instances with high contribution (\cref{sec:estimation}).
\end{example}

\begin{remark}
Current systems (e.g., JEVDB's $\varepsilon_r=\varepsilon_p=0.05$ relative to an LLM judge) provide judge-relative, operator-level, expected guarantees. The targets defined here are stated at the output level, either relative to the oracle or to a judge, and optionally with high probability.
\end{remark}

\section{Semantic Plan Equivalence}\label{sec:equivalence}

We define plan equivalence at four levels: oracle equivalence, implementation equivalence, confidence-relative equivalence, and admissibility. Approximate notions are defined relative to $\Qstar$ rather than between pairs of plans, since pairwise $\varepsilon$-closeness is not transitive.

\begin{definition}[Oracle equivalence]\label{def:oeq}
Logical plans $L_1\oeq L_2$ iff $L_1^*(\D)=L_2^*(\D)$ for all $\D$. By \cref{rem:classical} this coincides with classical equivalence of positive RA.
\end{definition}

\begin{definition}[Implementation equivalence]\label{def:impl-eq}
For a fixed PD policy $\Pol$, let $\shat$ denote the implemented predicate. $P_1=(L_1,\Pol)$ and $P_2=(L_2,\Pol)$ are implementation-equivalent iff $L_1[\shat](\D)=L_2[\shat](\D)$ for all $\D$.
\end{definition}

\begin{proposition}[Plan invariance under PD policies]\label{prop:invariance}
If $\Pol$ is pointwise-deterministic, then $L_1\oeq L_2$ implies that $P_1$ and $P_2$ are implementation-equivalent, and hence have identical output quality on every $\D$.
\end{proposition}
\begin{proof}
We have $\shat(\kappa)=1$ iff $\Pol$ accepts $\kappa$, or escalates it and the escalation model, itself a fixed function of $\kappa$, accepts it. Thus $\shat$ is a deterministic function of $\kappa$, and the plan computes $Q$ under the predicate interpretation $\shat$. Classical rewrites of positive RA are valid for any deterministic interpretation of the predicates, so $L_1[\shat](\D)=L_2[\shat](\D)$.
\end{proof}

\begin{example}[Pushdown under a PD policy]\label{ex:invariance}
With ``\act{Accept} iff $p\ge0.5$'', the pushdown plan of \cref{ex:instances} accepts $r_1,r_2,r_3$ and returns their $47$ orders. The pull-up plan evaluates $150$ instances, but each of the $100$ instances of $r_4$ carries the same text and the same $p=0.3$ and is rejected each time; it returns the same $47$ orders. The two plans differ in cost ($5$ vs.\ $150$ model calls) but return the same result. For PD policies, plan selection therefore reduces to classical cost-based optimization.
\end{example}

\begin{corollary}[Source of coupling]\label{cor:coupling}
Output quality can differ between oracle-equivalent plans only through (a) population-calibrated policies, (b) context-aware policies, (c) approximate (non-sound) pruning, or (d) model nondeterminism. In particular, JEVDB-Flash's accuracy is invariant to pushdown, while that of its quality-controlled mode is not, since its escalation band is fit on a sample of the candidates its plan presents. JEVDB motivates pushdown by cost considerations (\citealp[§3.1]{jevdb2026}). \Cref{prop:pc-unsound} below shows that, in its calibrated mode, pushdown also affects accuracy.
\end{corollary}

\paragraph{Confidence-relative equivalence.}
If plans can differ in quality, an optimizer needs to know their quality in order to compare them. Ground truth is not available at optimization time, but calibrated confidences provide the predicted quality of every candidate plan (\cref{def:predicted}), which an optimizer can use for this purpose.

\begin{definition}[Confidence-relative equivalence and dominance]\label{def:conf-eq}
Fix a calibrated decision model $M$. Plans $P_1,P_2$ with $P_1\oeq P_2$ (on their logical parts) are \emph{$M$-equivalent} on $\D$, written $P_1\equiv_M P_2$, if $\widehat{\Rec}_M(P_1)=\widehat{\Rec}_M(P_2)$ and $\widehat{\Prec}_M(P_1)=\widehat{\Prec}_M(P_2)$. $P_1$ \emph{$M$-dominates} $P_2$ if it is at least as good on both and no more expensive in expectation.
\end{definition}

\begin{example}\label{ex:conf-eq}
In \cref{tab:toy}, escalating $r_3$ or escalating $r_4$ costs the same one LLM call, but the predicted output quality is $(0.58,0.91)$ for $r_3$ and $(0.998,0.92)$ for $r_4$. The second plan $M$-dominates the first. Without confidences, the two plans cannot be told apart before execution, as both consist of the threshold plan plus one escalation.
\end{example}

\Cref{fig:toyspace} shows all physical plans of the running example by cost and predicted quality, and \cref{tab:eq-levels} summarizes the four levels of equivalence.

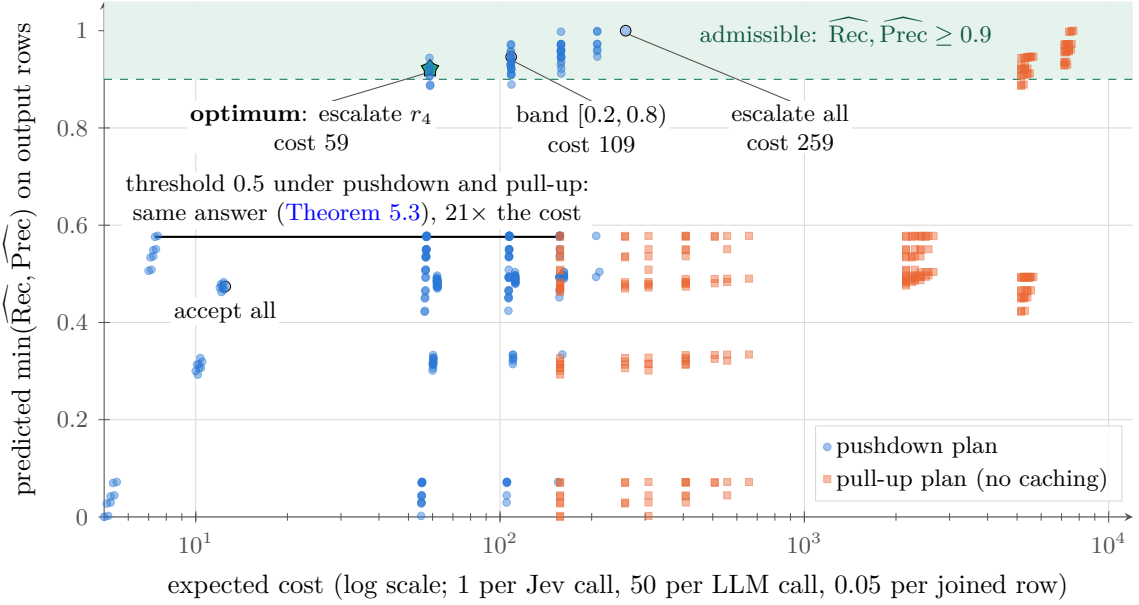
\begin{figure}[t]
\centering
\begin{tikzpicture}
\begin{axis}[paper, width=0.92\linewidth, height=8.4cm,
  xmode=log, xmin=5, xmax=12000, ymin=0, ymax=1.06,
  xlabel={expected cost (log scale; $1$ per Jev call, $50$ per LLM call, $0.05$ per joined row)},
  ylabel={predicted $\min(\widehat{\Rec},\widehat{\Prec})$ on output rows},
  legend style={at={(0.99,0.03)},anchor=south east,fill=white,fill opacity=0.9,text opacity=1,draw=cGrid}]
  \addplot[draw=none,fill=cAdm,forget plot] coordinates {(5,0.9) (12000,0.9) (12000,1.06) (5,1.06)} \closedcycle;
  \draw[cAqua!70!black,dashed] (axis cs:5,0.9) -- (axis cs:12000,0.9);
  \node[font=\footnotesize,color=cAqua!50!black,anchor=north west] at (axis cs:420,1.05) {admissible: $\widehat{\Rec},\widehat{\Prec}\ge0.9$};
  \addplot[only marks, mark=*, mark size=1.4pt, color=cBlue, fill opacity=0.45, draw opacity=0.6]
    table[x=cost, y=q, restrict expr to domain={\thisrow{plan}}{0:0}] {sim/data/toy_space.dat};
  \addplot[only marks, mark=square*, mark size=1.3pt, color=cOrange, fill opacity=0.45, draw opacity=0.6]
    table[x=cost, y=q, restrict expr to domain={\thisrow{plan}}{1:1}] {sim/data/toy_space.dat};
  \legend{pushdown plan, pull-up plan (no caching)}
  \draw[black,thick] (axis cs:7.35,0.576) -- (axis cs:157.5,0.576);
  \node[font=\footnotesize,anchor=south,align=center,fill=white,inner sep=1pt] at (axis cs:34,0.592)
    {threshold $0.5$ under pushdown and pull-up:\\same answer (\cref{prop:invariance}), $21\times$ the cost};
  \draw[cInk] (axis cs:28,0.855) -- (axis cs:58.85,0.915);
  \draw[cInk] (axis cs:200,0.855) -- (axis cs:108.75,0.94);
  \draw[cInk] (axis cs:900,0.855) -- (axis cs:258.6,0.995);
  \node[draw=black,fill=cAqua,star,star points=5,inner sep=1.6pt] at (axis cs:58.85,0.922) {};
  \node[draw=black,fill=white,circle,inner sep=1.4pt] at (axis cs:108.75,0.946) {};
  \node[draw=black,fill=white,circle,inner sep=1.4pt] at (axis cs:258.6,1.0) {};
  \node[draw=black,fill=white,circle,inner sep=1.4pt] at (axis cs:12.5,0.474) {};
  \node[font=\footnotesize,anchor=north,align=center,fill=white,inner sep=1pt] at (axis cs:24,0.855) {\textbf{optimum}: escalate $r_4$\\cost $59$};
  \node[font=\footnotesize,anchor=north,align=center,fill=white,inner sep=1pt] at (axis cs:200,0.855) {band $[0.2,0.8)$\\cost $109$};
  \node[font=\footnotesize,anchor=north,align=center,fill=white,inner sep=1pt] at (axis cs:900,0.855) {escalate all\\cost $259$};
  \node[font=\footnotesize,anchor=north,fill=white,inner sep=1pt] at (axis cs:12.5,0.45) {accept all};
\end{axis}
\end{tikzpicture}
\caption{The physical plan space of the running example, predicted from confidences alone. Each point is one plan: a logical plan (pushdown or pull-up) and an action per review ($3^5\times2=486$ plans). \emph{Equivalence:} all points are oracle-equivalent (same $\Qstar$); a PD policy gives two points at the same height, the same answer at different cost (black segment); points at the same height are $M$-equivalent (\cref{def:conf-eq}). \emph{Optimization:} the admissible plans are the shaded band, and cost--accuracy optimization picks its leftmost point. A classical optimizer with a fixed threshold only chooses between the two ends of the black segment and keeps the cheaper one, at predicted quality $0.58$; the operator-level band is admissible but costs $1.8\times$ the optimum.}\label{fig:toyspace}
\end{figure}

\begin{table}[t]
\centering
\caption{Four levels of plan equivalence, from strongest to weakest. Only the last two levels distinguish plans of different quality, and checking them before execution requires calibrated confidences.}\label{tab:eq-levels}
\small
\begin{tabularx}{\linewidth}{@{}lYYl@{}}
\toprule
Level & $P_1$ and $P_2$ are related iff & In \cref{fig:toyspace} & Checkable without labels? \\
\midrule
Oracle ($\oeq$) & same answer under oracle predicates & all $486$ points & yes (classical rules) \\
Implementation & same answer under the implemented predicates & two ends of a black segment (PD policy) & yes, for PD policies \\
Confidence-relative ($\equiv_M$) & same predicted output recall and precision & points at the same height & yes, from confidences \\
Admissibility ($\Adm$) & both meet the output target & points in the shaded band & predicted form ($\Adm_M$) only \\
\bottomrule
\end{tabularx}
\end{table}

\begin{definition}[Admissibility]\label{def:admissible}
$P$ is \emph{$T$-admissible} for $Q$ if it satisfies $T$ (\cref{def:target}) on $\D$. The admissible set $\Adm(Q,T,\D)$ takes the role of the equivalence class of correct plans. Its predicted counterpart $\Adm_M(Q,T,\D)$ uses $\widehat{\Rec}_M$ and $\widehat{\Prec}_M$.
\end{definition}

\begin{definition}[Quality-sound rewrite]\label{def:sound}
A rewrite $r$ is \emph{quality-sound} for a policy class $\PolClass$ if for every $P$ with policy in $\PolClass$, $P\in\Adm$ implies $r(P)\in\Adm$.
\end{definition}

\begin{proposition}\label{prop:pd-sound}
Every classical rewrite of positive RA is quality-sound for PD policies.
\end{proposition}
\begin{proof}
By \cref{prop:invariance}, $r(P)$ returns the same output as $P$ on every $\D$.
\end{proof}

\begin{proposition}[Pushdown is not quality-sound for PC policies]\label{prop:pc-unsound}
There is a query, database, target, and population-calibrated policy such that the pull-up plan is admissible and the pushdown plan is not.
\end{proposition}
\begin{proof}
Consider the query of the running example and the target $T=(0.9,0.9)$ in its expected variant. Let \textit{Reviews} contain $90$ reviews of niche products with $p=0.98$ and fan-out $1$, and $10$ reviews of popular products with $p=0.80$ and fan-out $100$. The policy is calibrated on the instances presented by its plan. It accepts every instance with $p\ge h$, where $h$ is the smallest score at which the expected precision of the accepted instances reaches $0.95$, and escalates the rest to a perfect judge. No instance is rejected, so recall is $1$ for both plans.
\begin{itemize}[leftmargin=*,nosep]
  \item \emph{Pushdown} calibrates on the $100$ reviews. Accepting all of them has expected precision $(90\cdot0.98+10\cdot0.80)/100=0.962\ge0.95$, so $h=0.80$ and nothing is escalated. The expected precision on output rows is $(88.2+800)/(90+1000)=0.815<0.9$, so the plan is \emph{not admissible}.
  \item \emph{Pull-up} calibrates on the $1{,}090$ joined rows. Accepting all of them has expected precision $0.815<0.95$, so $h=0.98$ and the $1{,}000$ rows of popular products are escalated. The judge keeps only the true ones, so the expected output precision is $888.2/(90+800)=0.998$ and the plan is \emph{admissible}.
\end{itemize}
The two plans calibrate on different populations, namely reviews for pushdown and joined rows (i.e., reviews weighted by fan-out) for pull-up. The popular reviews are $10\%$ of the first population and $92\%$ of the second. A PD policy with the pull-up plan's band would give both plans the same answer (\cref{prop:invariance}).
\end{proof}

\Cref{sec:experiments} reproduces this effect on a $2{,}000$-review workload in which the calibration sample is drawn at random.

\begin{proposition}[Sound and approximate pruning]\label{prop:pruning}
Call a pruning step \emph{sound} if it removes only tuples $t$ with $\sstar(t)=0$ (e.g., a lossless semantic Bloom filter or a relational Bloom semijoin). Under a PD or context-aware policy whose actions on the surviving instances are unchanged, a sound pruning step never decreases recall and never decreases precision. An approximate step that removes each tuple $t$ with $\sstar(t)=1$ independently with probability $\beta_t$ lowers the expected number of true positives by at most $\sum_t\beta_t\,q_t\,c^{\Pol}(t)$, with equality when no output's witness contains two such tuples (e.g., one semantic operator); here $q_t$ and $c^{\Pol}(t)$ are as in \cref{def:contribution}.
\end{proposition}
\begin{proof}
A sound step removes no tuple that occurs in a derivation of $\Qstar(\D)$. Hence true positives are unaffected, and the removed tuples could only have produced false positives. For the approximate step, a true positive is lost only if some tuple of its witness is removed. By a union bound over the removed tuples and linearity of expectation, the expected loss is at most $\sum_t\beta_t$ times the expected number of true positives whose witness contains $t$, which is $q_t\,c^{\Pol}(t)$. If no witnesses are shared, the events are disjoint and the bound is exact.
\end{proof}
Under a PC policy, pruning also changes the calibration population, so \cref{prop:pc-unsound} applies to it as well.

\begin{example}[Non-transitivity of pairwise closeness]\label{ex:nontransitive}
Measure closeness by
\[
d(P_i,P_j)=\frac{|P_i(\D)\,\triangle\,P_j(\D)|}{|\Qstar(\D)|}
\]
and take $\Qstar(\D)$ to be the $142$ orders of \cref{ex:oracle}. Let $P_1$ return all $142$, $P_2$ drop $7$ of them, and $P_3$ drop $7$ more. Then $d(P_1,P_2)=d(P_2,P_3)=7/142<0.05$, yet $\Rec(P_3)=128/142=0.90$, which violates $\rho=0.93$ although $P_1$ is perfect. A sequence of rewrites, each of which is close to the previous plan, can thus leave the admissible set. For this reason, \cref{def:admissible} defines correctness relative to $\Qstar$.
\end{example}

\section{Error Propagation via Lineage}\label{sec:propagation}

In \cref{def:predicted}, the expected error of each decision is weighted by its contribution $c_i$ to the output. In this section, we define $c_i$ for plans with several semantic operators and for both base and intermediate tuples. We then use contributions in two directions: forward, to compute output quality from per-decision confidences, and backward, to derive per-tuple prices from an output target. Since false negatives cannot be recovered by later operators while false positives may still be filtered out, recall and precision compose differently, and they favor different operator orders.

\subsection{Lineage and contribution}

\begin{definition}[Decision lineage]\label{def:lineage}
For $o\in\Qstar(\D)$ under bag semantics, its \emph{witness} $w(o)$ is the set of decision instances evaluated on the derivation of $o$. Under set semantics, $o$ has a set of witnesses $W(o)$ and its lineage is a DNF over Boolean decision variables $x_i$, i.e., its provenance polynomial in the Boolean semiring~\citep{green2007provenance}:
\[
\Phi_o=\bigvee_{w\in W(o)}\;\bigwedge_{i\in w}x_i .
\]
\end{definition}

\begin{definition}[Output contribution of a decision instance]\label{def:contribution}
Contributions are defined for decision instances, which are the units a policy acts on. An instance may refer to a base tuple (a semantic filter on a base relation) or to an intermediate tuple (a join pair, or a filter above a join). We distinguish two versions.
\begin{itemize}[leftmargin=*]
  \item \emph{Oracle contribution} $c^*(i)$: the number of tuples of $\Qstar(\D)$ whose witness contains $i$. It depends only on $Q$ and $\D$, not on the plan, and we use it in definitions.
  \item \emph{Policy contribution} $c^{\Pol}(i)$: the expected number of correctly returned outputs lost if $i$ is answered negatively, given the other decisions under $\Pol$. This is the quantity used by the optimizer:
  \[
  c^{\Pol}(i)=\E_{\Pol}[\TP\mid x_i=1]-\E_{\Pol}[\TP\mid x_i=0]\;\overset{\text{bag SPJ}}{=}\;\sum_{o\,:\,i\in w(o)}\;\prod_{j\in w(o)\setminus\{i\}}q_j ,
  \]
  where $q_j$ is the probability that decision $j$ is a correct positive under $\Pol$ ($p_j$ if $j$ is accepted or escalated to a perfect judge, $0$ if rejected).
\end{itemize}
\end{definition}

With a single semantic operator, all other factors equal $1$, and $c^{\Pol}(r)=f_r$ is the fan-out used in \cref{sec:quality}. With several operators, downstream decisions reduce the contribution of upstream decisions, and upstream decisions determine whether a downstream decision affects the output at all.

\begin{example}[Base and intermediate decisions]\label{ex:two-ops}
We extend the running example with a second predicate that reads both tables and is therefore evaluated on joined rows: \texttt{SEM((r.text, o.notes), 'the customer reports the same defect')}. Its decision instances are intermediate tuples $(r,o)$. For review $r_4$ ($p=0.3$, $100$ orders), suppose $20$ of its joined rows get confidence $p=0.9$ and $80$ get $p=0.05$, and the second operator accepts iff $p\ge0.5$; for the accepted rows, $q_o=p_o=0.9$.
\begin{itemize}[leftmargin=*,nosep]
  \item \emph{Output to base.} An order of $r_4$ is a true answer only if both predicates hold. Rejecting $r_4$ loses the $20$ accepted rows, each a true positive with probability $0.9$, so $c^{\Pol}(r_4)=20\cdot0.9=18$ rather than $100$. The expected number of output false negatives caused by rejecting $r_4$ is $0.3\cdot18=5.4$ instead of $30$; because the downstream predicate is selective, an error on $r_4$ costs less.
  \item \emph{Output to intermediate.} A joined row $(r_4,o)$ can only contribute if $r_4$ is kept, i.e., $c^{\Pol}((r_4,o))=q_{r_4}=0.3$ if $r_4$ is accepted and $0$ if it is rejected. If $r_4$ is rejected, none of the $100$ decisions of the second operator on its rows can affect the output, and evaluating them is wasted work; this is the usual argument for pushing the first filter down. If $r_4$ is escalated and confirmed, each of these decisions matters fully.
\end{itemize}
\end{example}

\Cref{fig:propagation} illustrates both cases.

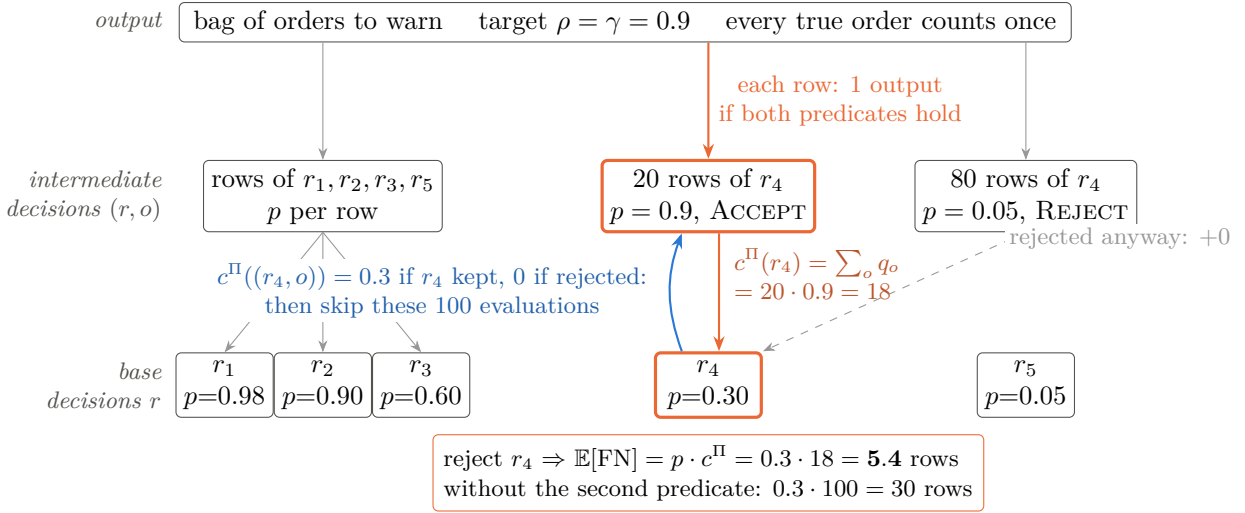
\begin{figure}[t]
\centering
\begin{tikzpicture}[
  font=\small,
  box/.style={draw=cInk,rounded corners=2pt,align=center,inner sep=3pt,fill=white},
  hot/.style={box,draw=cOrange,very thick},
  down/.style={-{Stealth[length=2mm]},thick,cOrange},
  up/.style={-{Stealth[length=2mm]},thick,cBlue},
  faint/.style={-{Stealth[length=1.6mm]},cInk!60},
  lab/.style={font=\footnotesize,align=center,fill=white,inner sep=1.5pt}]
\node[font=\footnotesize\itshape,color=cInk,anchor=east] at (-1.1,0) {output};
\node[font=\footnotesize\itshape,color=cInk,anchor=east,align=right] at (-1.1,-2.3) {intermediate\\decisions $(r,o)$};
\node[font=\footnotesize\itshape,color=cInk,anchor=east,align=right] at (-1.1,-4.8) {base\\decisions $r$};
\node[box,minimum width=11.8cm] (out) at (4.9,0) {bag of orders to warn \quad target $\rho=\gamma=0.9$ \quad every true order counts once};
\node[box] (rest) at (0.9,-2.3) {rows of $r_1,r_2,r_3,r_5$\\$p$ per row};
\node[hot] (a) at (6.0,-2.3) {$20$ rows of $r_4$\\$p=0.9$, \act{Accept}};
\node[box] (b) at (10.2,-2.3) {$80$ rows of $r_4$\\$p=0.05$, \act{Reject}};
\foreach \n/\p/\x in {r_1/0.98/-0.4, r_2/0.90/0.9, r_3/0.60/2.2} {
  \node[box,minimum width=1.05cm] (\n) at (\x,-4.8) {$\n$\\$p{=}\p$};
}
\node[hot,minimum width=1.4cm] (r4) at (6.0,-4.8) {$r_4$\\$p{=}0.30$};
\node[box,minimum width=1.05cm] (r5) at (10.2,-4.8) {$r_5$\\$p{=}0.05$};
\draw[down] (out.south -| a) -- node[lab,right,xshift=2pt]{each row: $1$ output\\if both predicates hold} (a);
\draw[faint] (out.south -| b) -- (b);
\draw[faint] (out.south -| rest) -- (rest);
\draw[down] ([xshift=4pt]a.south) -- ([xshift=4pt]r4.north);
\node[lab,anchor=north west,align=left,text=cOrange!80!black] at ([xshift=8pt,yshift=-4pt]a.south) {$c^{\Pol}(r_4)=\sum_o q_o$\\$=20\cdot0.9=18$};
\draw[faint,dashed] (b.south) -- node[lab,pos=0.12,right,xshift=4pt,yshift=3pt]{rejected anyway: $+0$} (r4.north east);
\foreach \n in {r_1,r_2,r_3} \draw[faint] (rest.south) -- (\n.north);
\draw[up] ([xshift=-10pt]r4.north) to[out=110,in=-110] ([xshift=-10pt]a.south);
\node[lab,anchor=east,text=cBlue!80!black] at (5.25,-3.55) {$c^{\Pol}((r_4,o))=0.3$ if $r_4$ kept, $0$ if rejected:\\then skip these $100$ evaluations};
\node[draw=cOrange,rounded corners=2pt,align=left,font=\footnotesize,anchor=north,inner sep=4pt] at (6.0,-5.45)
  {reject $r_4$ $\Rightarrow$ $\E[\FN]=p\cdot c^{\Pol}=0.3\cdot18=\mathbf{5.4}$ rows\\
   without the second predicate: $0.3\cdot100=30$ rows};
\end{tikzpicture}
\caption{Propagation of contributions from the output to intermediate and base decisions (\cref{ex:two-ops}). Orange: the contribution of a decision is the expected number of true outputs it contributes to, computed top-down. Because the second predicate is selective, the contribution of $r_4$ is $18$ rather than $100$, and rejecting $r_4$ loses $5.4$ instead of $30$ expected rows. Blue: an intermediate decision only matters if its base tuple is kept, so if $r_4$ is rejected, the $100$ decisions of the second predicate on its rows need not be evaluated.}\label{fig:propagation}
\end{figure}

Note that $c^{\Pol}(i)$ is the influence of $x_i$ on expected output quality in the sense of sensitivity analysis for probabilistic databases~\citep{kanagal2011sensitivity}. Tuple Shapley values~\citep{livshits2020shapley} would be an alternative way to attribute output quality to decisions.

\subsection{From decision errors to output error}

Under bag semantics and independent decisions, each expected output count is a sum over output tuples of a product over their witnesses. For instance,
\[
\E[\TP]=\sum_{o}\prod_{j\in w(o)}q_j .
\]
For one semantic operator this is \cref{def:predicted},
\[
\E[\FN]=\sum_{i\,\text{rejected}}p_i\,c^{\Pol}(i),\qquad \E[\FP]=\sum_{i\,\text{accepted}}(1-p_i)\,c^{\Pol}(i),
\]
i.e., the output error is the sum of the per-decision errors $1-\mathrm{conf}_i$ weighted by $c^{\Pol}(i)$. With several operators, the expected number of false negatives caused by rejecting $i$ alone is $p_i\,c^{\Pol}(i)$, as in \cref{ex:two-ops}. Under set semantics, which we discuss in \cref{sec:estimation}, the sum of products is replaced by the probability of a DNF.

\paragraph{Computing contributions by backward propagation.}
Computing $c(i)$ requires propagating weight from the output back to every decision instance. For an acyclic join tree under bag semantics, this can be done by sum-product message passing: a top-down pass sends each tuple the expected number of surviving completions in the rest of the tree, where semantic operators contribute factors $q_j$ (or selectivity estimates for operators not yet evaluated). This takes linear time in the size of the pass~\citep{abokhamis2016faq} and computes $c$ for base and intermediate instances simultaneously. In \cref{ex:two-ops}, the message sent from the orders of $r_4$ to $r_4$ is $\sum_o q_o\,[o\text{ accepted}]=18$. The cost saved by dropping $i$, $\Delta\Cost(i)$, can be computed in the same way in a cost semiring, as the expected number of downstream rows times the per-row cost of each downstream operator.

\subsection{Pushing an output target down to tuples}

An output-level target $(\rho,\gamma)$ constrains sums of per-decision terms, which suggests enforcing it through a price. Suppose each unit of expected output error has price $\lambda$. Resolving decision $i$, e.g., by escalation, removes $c^{\Pol}(i)\,(1-\mathrm{conf}_i)$ expected output errors at the cost of the escalation. Dropping a tuple saves its downstream cost but adds $c^{\Pol}(i)\,p_i$ expected false negatives. Each tuple can then decide locally by comparing these quantities. \Cref{prop:lagrange} makes this precise, using one price for the recall constraint and one for the precision constraint.

\begin{example}[Pricing the running example]\label{ex:price}
Let an escalation cost $50$ units and let each expected wrong output row cost $\lambda=10$ units. Escalating $r_4$ removes $30$ expected missing rows, worth $300>50$, so $r_4$ is escalated. Escalating $r_2$ would remove $4$ expected false rows (worth $40<50$) and escalating $r_3$ would remove $2$ (worth $20<50$), so both keep the decision of the cheap model. The resulting plan has predicted output quality $(0.998,0.92)$ and meets $T=(0.9,0.9)$ with a single LLM call. An operator-level band such as $[0.2,0.8)$ escalates both $r_3$ and $r_4$, and thus spends a second call on a review that accounts for only $2$ expected errors. With the second predicate of \cref{ex:two-ops}, the benefit of escalating $r_4$ drops to $5.4\cdot10=54$, which is still slightly above $50$; with a somewhat lower price, $r_4$ would be left to the cheap model.
\end{example}

\paragraph{Caveats.}
\begin{enumerate}[leftmargin=*]
  \item $c^{\Pol}$ is a marginal quantity. Contributions are not additive if two dropped instances share a witness, so selecting a set of instances to drop is a set-function problem with a coverage-type loss rather than a plain knapsack problem (\cref{sec:optimization}).
  \item When an upstream decision is made, downstream decisions are not yet known, so $c^{\Pol}$ must be estimated. Policy and contributions are thus defined by a fixed point.
  \item Under set semantics $c^{\Pol}$ is the influence on a DNF probability, \#P-hard in general (\cref{sec:estimation}).
  \item The denominator $|\Qstar(\D)|$ is the same for all instances and does not affect their ranking.
\end{enumerate}

\subsection{Recall and precision compose differently}

\begin{proposition}[Recall composes multiplicatively]\label{prop:recall}
For a bag SPJ query, suppose that every true-positive instance of operator $j$ is accepted with probability $R_j$, independently across instances. Then every true output survives with probability $\prod_j R_j$, so $\E[\Rec]=\prod_j R_j$. Consequently, a global target $\rho$ is allocated log-additively:
\[
\textstyle\sum_j\log R_j\ge\log\rho .
\]
\end{proposition}
\begin{proof}
A true output survives iff every decision on its witness accepts its instance, all of which are true positives. By independence this has probability $\prod_jR_j$ for every true output, and $\E[\Rec]$ is the average over true outputs.
\end{proof}
For example, two filters with recall $0.95$ each yield an output recall of $0.9025$, so an output target of $0.95$ requires a recall of about $0.975$ per operator. The product form no longer holds when recall varies across instances with different $c(i)$; in \cref{tab:toy}, an operator recall of $0.88$ results in an output recall of $0.58$.

For an accepted instance $i$, let $d^{\Pol}(i\mid\neg)$ be the expected number of candidate output rows whose witness contains $i$ and which every \emph{other} decision on the witness accepts, conditioned on $\sstar(i)=0$.

\begin{proposition}[Precision is order- and correlation-dependent]\label{prop:precision}
Under bag semantics, accepting instance $i$ instead of rejecting it adds $(1-p_i)\,d^{\Pol}(i\mid\neg)$ expected false output rows. In particular, if the other decisions on $i$'s rows reject them whenever $\sstar(i)=0$, accepting $i$ costs no output precision, however uncertain $i$ is.
\end{proposition}
\begin{proof}
If $\sstar(i)=0$, which has probability $1-p_i$, every returned row whose witness contains $i$ is a false positive, and the rows returned are exactly those accepted by all other decisions on their witness. Linearity of expectation gives the first claim; the second is the case $d^{\Pol}(i\mid\neg)=0$.
\end{proof}
Since $d^{\Pol}(i\mid\neg)$ depends on the decisions on $i$'s witnesses and on how they treat $i$'s false positives, the precision cost of a threshold depends on the operator order and on correlations between predicates. This is in contrast to recall (\cref{prop:recall}).
For example, if most reviews that wrongly pass the fire filter concern unrelated defects, the second predicate of \cref{ex:two-ops} rejects their rows anyway, and a looser threshold for the first filter saves escalations with little loss in output precision. This suggests two heuristics:
\begin{itemize}[leftmargin=*]
  \item \emph{Recall pushdown}: the recall budget should be spent on instances with small contribution $c(i)$.
  \item \emph{Precision pull-up}: operators upstream of strict, independent filters may run with lower precision.
\end{itemize}

\paragraph{Computing fan-out.}
For acyclic relational subqueries, per-tuple witness counts can be computed in linear time with counting semijoins (Yannakakis-style~\citep{yannakakis1981} bottom-up and top-down passes with counts). $c(i)$ additionally requires downstream decision probabilities (\cref{sec:estimation}).

\section{Estimating Quality from Calibrated Confidences}\label{sec:estimation}

With calibrated scores, each semantic decision can be modeled as a Bernoulli variable, and the expected quality of a plan becomes a query over a probabilistic database. \Cref{def:predicted} covered a single operator under bag semantics. In this section, we consider general SPJ plans, set semantics, imperfect judges, and miscalibrated models.

\paragraph{Output level.}
We treat the decision variables $x_i$ as independent, with $\Pr[x_i]=q_i$ the probability that $i$ is a correct positive under the policy's action ($p_i$ on accept, $0$ on reject, $\alpha_J\,p_i$ on escalation to a judge of recall $\alpha_J$). Then $\Pr[o\text{ returned correctly}]=\Pr[\Phi_o]$, which amounts to query evaluation over a tuple-independent probabilistic database~\citep{dalvi2007dichotomy,suciu2011probdb}.

\begin{definition}[Lineage-inducing query]\label{def:lineage-query}
For a plan whose logical part is a self-join-free SPJ query $Q$, the \emph{lineage-inducing query} $Q^{\mathrm{lin}}$ is the conjunctive query obtained by replacing every semantic operator $s$ over relations $R_1,\dots,R_k$ by a fresh probabilistic relation $S_s(\kappa)$ whose tuple $\kappa$ has probability $q_\kappa$. Relational atoms are deterministic.
\end{definition}

The lineage of an output tuple $o$ in the plan coincides with its lineage in $Q^{\mathrm{lin}}$, so $\Pr[o\text{ returned correctly}]=\Pr[o\in Q^{\mathrm{lin}}(\D)]$.

\begin{example}[Bag vs.\ set semantics]\label{ex:set}
Suppose the running example returns distinct customers (\texttt{SELECT DISTINCT} \texttt{o.customer}) and that customer Ann bought both the product of $r_2$ ($p=0.9$) and that of $r_4$ ($p=0.3$), both of which are accepted. Under bag semantics, Ann appears twice, and the expected number of correct rows for Ann is the sum $0.9+0.3=1.2$. Under set semantics, Ann is returned correctly iff $x_2\vee x_4$ holds, which has probability $1-(1-0.9)(1-0.3)=0.93$. Sums of products remain easy to compute for larger queries, whereas probabilities of DNFs in general do not.
\end{example}

\begin{theorem}[Complexity of quality evaluation]\label{thm:hardness}
Let $P$ be a plan over a self-join-free SPJ query with semantic selections and calibrated, independent decisions.
\begin{enumerate}[leftmargin=*,nosep]
  \item Under bag semantics, $\E[\TP]$, $\E[\FP]$, and $\E[\FN]$ are computable in time polynomial in $|\D|$.
  \item Under set semantics, $\E[\TP]$ is computable in polynomial time if $Q^{\mathrm{lin}}$ is hierarchical. If every relation of $Q$ carries a semantic selection (so every atom of $Q^{\mathrm{lin}}$ is probabilistic) and $Q^{\mathrm{lin}}$ is not hierarchical, it is \#P-hard.
\end{enumerate}
\end{theorem}
\begin{proof}
(1) By linearity of expectation, $\E[\TP]=\sum_o\prod_{j\in w(o)}q_j$, which is a sum over derivations of products of independent factors. It can be computed by evaluating the join in the sum-product semiring, i.e., by multiplying a probability column along each derivation and summing at the end. This takes polynomial time, and linear time for acyclic queries~\citep{abokhamis2016faq}. The cases of $\E[\FP]$ and $\E[\FN]$ are analogous, with factor $1-p_j$ or $p_j$ for the erroneous decision.

(2) For each output $o$, $\Pr[o\text{ returned correctly}]=\Pr[o\in Q^{\mathrm{lin}}(\D)]$, and $\E[\TP]$ is the sum of these probabilities. If $Q^{\mathrm{lin}}$ is hierarchical, each of these probabilities can be computed with a safe plan in polynomial time~\citep{dalvi2007dichotomy}. For hardness, let $q$ be a non-hierarchical self-join-free Boolean CQ, e.g., $q=\exists x,y\,R(x)\wedge S(x,y)\wedge T(y)$, over a tuple-independent database with probabilities $\mu(t)$; computing $\Pr[q]$ is \#P-hard~\citep{dalvi2007dichotomy}. We construct an SRA instance with the same relations, give each tuple $t$ a text attribute, and apply one semantic selection to each relation with calibrated confidence $p_t=\mu(t)$. Let the plan accept every decision (a PD policy) and project onto the empty attribute list under set semantics. The plan returns the single empty tuple whenever the relational join is nonempty, and that tuple is a true answer iff some derivation has all its oracle labels equal to $1$. Hence $\E[\TP]=\Pr[q]$. The same construction works for any non-hierarchical self-join-free $q$, with $Q^{\mathrm{lin}}=q$. (When some relations carry no semantic operator, their atoms are deterministic and the hardness boundary can move; we leave the exact dichotomy for that case open.)
\end{proof}

The running example under bag semantics falls under case (1), and every quantity in \cref{tab:toy,ex:predicted,ex:price} is a sum of products. Case (2) becomes relevant, e.g., for \texttt{DISTINCT} over a chain of two semantic joins, where sampling-based approximations of DNF probabilities can be used.

\paragraph{Imperfect judges.}
If the judge's confidence is set to $1$, all guarantees become judge-relative. Instead, we model the accuracy $\alpha_J$ of the judge, estimated on a small gold set, and set $\mathrm{conf}_i=\alpha_J$ for escalated decisions in \cref{def:predicted}. Escalated instances are, by design, the most difficult ones, so the judge's accuracy on them is likely below its average accuracy. It should therefore be estimated on the escalated band and, following \cref{ex:judge}, weighted by contribution.

\paragraph{Assumptions on calibration.}
\begin{itemize}[leftmargin=*]
  \item Calibration holds with respect to a distribution of instances. Pushdown or a join changes which instances are evaluated, and the resulting covariate shift can break calibration. In \cref{prop:pc-unsound}, the same reviews make up $10\%$ of the population of one plan and $92\%$ of the other.
  \item Errors may be correlated, since the same document appears in many pairs (and as a shared anchor in batched request layouts). This violates the independence assumption.
  \item Prediction-powered inference (PPI)~\citep{angelopoulos2023ppi} combines model predictions with a small gold sample to correct bias and obtain confidence intervals, and SUPG-style procedures~\citep{kang2020supg} provide precision and recall guarantees for proxy-based selection.
\end{itemize}

\paragraph{Contribution-weighted PPI.}
For an output-level total $\theta=\sum_i w_i\,y_i$ (e.g., $\TP$ with $w_i=c_i\,[i\text{ accepted}]$ and $y_i=\sstar(i)$), the predicted value $\sum_iw_ip_i$ is biased if $M$ is miscalibrated. With $n$ gold labels drawn with probability proportional to $c_i$, the PPI estimate
\[
\hat\theta=\sum_iw_i\,p_i+\frac{1}{n}\sum_{j\in\text{gold}}\frac{w_j\,(y_j-p_j)}{c_j/\sum_kc_k}
\]
is unbiased, and sampling in proportion to contribution reduces its variance when fan-out is skewed. In \cref{sec:experiments}, this estimator removes the precision bias of an overconfident model using $100$ labels.

\section{The Optimization Problem and Complexity}\label{sec:optimization}

We now state the optimization problem: choose a logical plan and a decision policy that minimize expected cost subject to an output-level quality target. With predicted quality, the constraint can be checked, and with contributions, it can be decomposed over tuples. The general problem nevertheless remains hard.

\begin{problem}[Cost--Accuracy Semantic Query Optimization, CASQO]\label{prob:casqo}
Given $Q$, $\D$, a target $T=(\rho,\gamma,\delta)$, a set of decision models with prices, and a policy class $\PolClass$,
\[
\min_{L\oeq Q,\;\Pol\in\PolClass}\;\E\big[\Cost((L,\Pol),\D)\big]\quad\text{s.t.}\quad (L,\Pol)\in\Adm(Q,T,\D).
\]
\end{problem}

\paragraph{Variants.}
\begin{itemize}[leftmargin=*]
  \item \emph{Operator-level allocation}: $\Pol$ assigns one threshold or band per semantic operator.
  \item \emph{Instance-level allocation}: $\Pol$ may treat instances of the same operator differently based on context (fan-out, downstream cost); the running example belongs to this variant.
  \item \emph{Judge-relative vs.\ oracle-relative} target; \emph{expected vs.\ high-probability} constraint.
\end{itemize}

\subsection{Per-tuple decomposition}

Fix the logical plan and consider the predicted variant of the expected target. Since $\widehat{\Rec}\ge\rho$ iff $(1-\rho)\,\E[\TP]-\rho\,\E[\FN]\ge0$, and $\widehat{\Prec}\ge\gamma$ iff $(1-\gamma)\,\E[\TP]-\gamma\,\E[\FP]\ge0$, both constraints are linear in the expected counts. Under bag semantics and with fixed contributions $c_i$, each expected count is a sum of per-instance terms (\cref{def:predicted}). Instance $i$ adds $c_ip_i$ to $\TP$ if accepted or escalated, $c_ip_i$ to $\FN$ if rejected or skipped, and $c_i(1-p_i)$ to $\FP$ if accepted. Let $\Cost_i(a)$ denote the expected cost of action $a$ on $i$, including the downstream work it causes.

\begin{proposition}[Per-tuple Lagrangian decomposition]\label{prop:lagrange}
For multipliers $\lambda_R,\lambda_P\ge0$, let each instance choose
\[
a_i(\lambda)\in\arg\min_{a}\;\Cost_i(a)+\lambda_R\big(\rho\,\FN_i(a)-(1-\rho)\,\TP_i(a)\big)+\lambda_P\big(\gamma\,\FP_i(a)-(1-\gamma)\,\TP_i(a)\big).
\]
If the resulting policy meets both linearized constraints, and each multiplier is zero unless its constraint is tight, then the policy is optimal among all instance-level policies for the fixed logical plan and predicted target.
\end{proposition}
\begin{proof}
Let $g_R(a)=(1-\rho)\E[\TP]-\rho\E[\FN]$ and $g_P(a)=(1-\gamma)\E[\TP]-\gamma\E[\FP]$. The Lagrangian $\Cost(a)-\lambda_Rg_R(a)-\lambda_Pg_P(a)$ is a sum of per-instance terms, so $a(\lambda)$ minimizes it over all action vectors. For any feasible $a'$, $\Cost(a')\ge\Cost(a')-\lambda_Rg_R(a')-\lambda_Pg_P(a')\ge\Cost(a(\lambda))-\lambda_Rg_R(a(\lambda))-\lambda_Pg_P(a(\lambda))=\Cost(a(\lambda))$, the last step by complementary slackness (Everett's sufficiency argument).
\end{proof}

In terms of \cref{fig:toyspace,fig:planspace}, CASQO asks for the leftmost point in the shaded region. Intuitively, \cref{prop:lagrange} charges each tuple for its own expected output errors, scaled by its contribution, which corresponds to the backward arrow in \cref{fig:chain}. An operator-level policy applies the same rule with all $c_i$ set to $1$ and thus treats a review with $2$ orders like a review with $100$. Since actions are discrete, multipliers $(\lambda_R,\lambda_P)$ that satisfy the conditions exactly need not exist. In that case, a grid search over the two multipliers that returns the cheapest feasible policy is a heuristic; this is the variant evaluated in \cref{sec:experiments}.

\begin{example}[Decomposition on the running example]\label{ex:lagrange}
With $\rho=\gamma=0.9$, an escalation cost of $50$, and row costs ignored, $\lambda_R=12$ and $\lambda_P=0$ yield the following values for $r_4$: reject costs $\lambda_R\rho\,c_4p_4=12\cdot0.9\cdot30=324$, accept costs $-\lambda_R(1-\rho)\cdot30=-36$ (accepting adds true positives, which the recall constraint rewards), and escalate costs $50-36=14$. Accepting minimizes this objective, so a price on precision is needed as well. With $\lambda_P=5$, accepting $r_4$ adds $5(0.9\cdot70-0.1\cdot30)=300$ for its $70$ expected false rows, which gives $264$ for accept, $324$ for reject, and $50-36-15=-1$ for escalate, so $r_4$ is escalated. For $r_3$ ($c=5$, $p=0.6$) the same prices give $3.9$ for accept, $32.4$ for reject, and $44.9$ for escalate; similarly, $r_1$ and $r_2$ are accepted and $r_5$ is rejected. The resulting policy escalates $r_4$ only, as in \cref{ex:price}. Since the resulting recall is $0.998$, the recall constraint is not tight, and \cref{prop:lagrange} does not certify optimality. The policy is nevertheless optimal, since every feasible policy needs at least one LLM call (accepting $r_4$ would reduce predicted precision to $0.48$).
\end{example}

\subsection{Confidence-centric skipping}\label{sec:skip}

So far we assumed that the decision model scores every instance before the policy acts. This is the usual assumption of proxy-based query processing, where a cheap proxy is applied to all inputs~\citep{kang2017noscope,kang2020supg,kang2021abae}, and the cost of these calls is part of $\Cost$. It is not necessary, however. During execution, the decisions taken so far may already guarantee the target, regardless of how the remaining instances turn out. From that point on, the remaining instances can be skipped without being scored.

This differs from data skipping with min/max summaries such as small materialized aggregates~\citep{moerkotte1998sma}, i.e., zone maps. There, a block is skipped because its value range shows that no tuple in it can satisfy the predicate, and the result is unchanged. Here, an instance may well satisfy the predicate; it is skipped because the confidence accumulated on the instances processed so far shows that the output target holds whatever the instance turns out to be. We therefore call it \emph{confidence-centric skipping}. The result may change, but the target is guaranteed.

Skipping an instance never returns its rows, so it cannot add false positives; it can only add false negatives, at most one per row. The fan-out of an unprocessed instance is known from the relational part of the plan, e.g., by counting semijoins (\cref{sec:propagation}), without evaluating any semantic predicate. Let $D$ be the instances processed so far, with expected counts $\E[\TP_D]$, $\E[\FP_D]$, $\E[\FN_D]$ as in \cref{def:predicted}, and let $U$ be the unprocessed instances. For each $i\in U$, let $u_i\ge p_i$ be an upper bound on its confidence that is available without scoring $i$; $u_i=1$ always qualifies.

\begin{proposition}[Confidence-centric skipping]\label{prop:skip}
Under bag semantics, if
\[
\E[\TP_D]\;\ge\;\rho\Big(\E[\TP_D]+\E[\FN_D]+\sum_{i\in U}c_i\,u_i\Big)\qquad\text{and}\qquad \E[\TP_D]\;\ge\;\gamma\big(\E[\TP_D]+\E[\FP_D]\big),
\]
then the plan that skips every instance in $U$ meets the predicted target $(\rho,\gamma)$ for every assignment of confidences $p_i\le u_i$ to the instances in $U$.
\end{proposition}
\begin{proof}
Skipping $U$ adds nothing to $\TP$ or $\FP$ and adds $\sum_{i\in U}c_ip_i\le\sum_{i\in U}c_iu_i$ to $\E[\FN]$. Hence $\widehat{\Rec}\ge\E[\TP_D]/(\E[\TP_D]+\E[\FN_D]+\sum_{i\in U}c_iu_i)\ge\rho$, and $\widehat{\Prec}=\E[\TP_D]/(\E[\TP_D]+\E[\FP_D])\ge\gamma$.
\end{proof}

With $u_i=1$, the guarantee holds for any confidences of the skipped instances, even if all of them were certain positives. The condition is checked before each instance is scored, so its cost is negligible.

\paragraph{Processing order.}
Whether and when the condition is reached depends on the order in which instances are processed. The order should be chosen from information that is available before scoring. With $u_i=1$, the bound depends on the unprocessed instances only through their residual mass $\sum_{i\in U}c_i$, which after any number $k$ of processed instances is smallest if the instances with the largest contributions come first. Decreasing contribution, here decreasing fan-out, is therefore the natural order; with skewed fan-out, a small head of the instances carries almost all output rows, and the tail never needs to be scored. In storage order, the residual mass shrinks only linearly and the condition is typically reached near the end, if at all. Ordering by contribution times a rough prior on the confidence can be worse, because the bound charges every unscored row as a potential true positive regardless of the prior; what matters is how fast the residual mass decreases. The order also determines how quickly a usable partial result is available and, under a population-calibrated policy, which instances enter the calibration sample.

\paragraph{Confidence summaries.}
Tighter bounds $u_i$ stop earlier. They can come from per-block summaries in the spirit of zone maps, storing, e.g., the maximum confidence of each block under a previously evaluated predicate or a reusable condition index~\citep{jevdb2026}; instances are then processed in decreasing order of $c_iu_i$. Such summaries help only if blocks group instances of similar confidence: if every block contains some confident instance, every block maximum is close to $1$ and nothing is gained. If no summary is available, an upper bound on the positive rate of the tail, e.g., from a small sample, can replace $u_i$; this stops earlier but gives a weaker, statistical guarantee.

\begin{example}[Skipping in the running example]\label{ex:skip}
Process the reviews of \cref{tab:toy} in decreasing order of fan-out, $r_4,r_2,r_3,r_5,r_1$, with the policy of \cref{ex:price}: escalate $r_4$ and accept $r_2$ and $r_3$. After $r_4$ and $r_2$, we have $\E[\TP_D]=30+36=66$, $\E[\FN_D]=0$, and $3+2+5=10$ unprocessed rows, so with $u_i=1$ the recall bound is $66/(66+10)=0.87<0.9$. After $r_3$, we have $\E[\TP_D]=69$, $\E[\FP_D]=4+2=6$, and $5$ unprocessed rows. The recall bound is $69/(69+5)=0.93$ and the precision is $69/75=0.92$, so $r_5$ and $r_1$ are skipped without being scored. The target holds whatever their confidences are, even if both were certain positives; with their actual confidences, the predicted recall is $69/(69+0.15+1.96)=0.97$. Now suppose a summary places $r_3$ and $r_5$ in a block whose maximum confidence is $0.6$ and $r_1$ in a block with maximum $0.98$. After $r_4$ and $r_2$, the residual bound is $5\cdot0.6+3\cdot0.6+2\cdot0.98=6.76$, the recall bound is $66/(66+6.76)=0.91$, and the precision is $66/70=0.94$, so already $r_3$ is skipped as well.
\end{example}

In our simulations, skipping leaves most instances unscored when fan-out is heavy-tailed, almost none in storage order, and even more with a summary clustered by confidence (\cref{sec:experiments}).

\subsection{Complexity}

\begin{proposition}[Instance-level dropping is NP-hard]\label{prop:np}
Choosing which instances to \act{Skip} so as to maximize saved cost subject to output recall $\ge\rho$ is NP-hard, even for one semantic filter followed by one join and with all confidences equal to $1$.
\end{proposition}
\begin{proof}
We reduce from 0/1 knapsack with weights $w_k$, values $v_k$, capacity $W$, and target value $V$. For each item, create a review $r_k$ with $p=1$ and fan-out $f_k=w_k$, and give its joined rows a downstream per-row operator whose total cost over $r_k$'s rows is $v_k$ (e.g., by the token length of those rows). Skipping $r_k$ saves $v_k$ and loses $w_k$ true output rows. With $\rho=1-W/\sum_kw_k$, a skip set is feasible iff its total weight is at most $W$, and it saves at least $V$ iff the knapsack instance is a yes-instance.
\end{proof}
Note that the reduction does not use uncertainty; hardness already arises from fan-out-weighted recall. With uncertain confidences, the weights are scaled by $p_k$ and the reduction still applies.

\begin{proposition}[Operator-level recall allocation is convex]\label{prop:convex}
Under the assumptions of \cref{prop:recall}, let operator $j$ reach recall $R_j$ at expected cost $C_j(R_j)$, convex and nondecreasing on $(0,1]$. Then minimizing $\sum_jC_j(R_j)$ subject to $\prod_jR_j\ge\rho$ is a convex program, and at an interior optimum $R_j\,C_j'(R_j)$ is the same for every operator.
\end{proposition}
\begin{proof}
The constraint is equivalent to $\sum_j\log R_j\ge\log\rho$. Since $\log$ is concave, the feasible set is a superlevel set of a concave function and hence convex; the objective is convex by assumption. The KKT conditions $C_j'(R_j)=\lambda/R_j$ give $R_jC_j'(R_j)=\lambda$.
\end{proof}
In words, the marginal cost of an additional unit of log-recall is the same for all operators, so operators that are cheap to make accurate should receive a larger share of the recall budget. With two filters and target $\rho=0.9$, equal cost curves give $R_1=R_2=\sqrt{0.9}\approx0.949$.

\Cref{tab:results} summarizes the results and their status.

\begin{table}[t]
\centering
\caption{Results and their status.}\label{tab:results}
\begin{tabularx}{\linewidth}{@{}lYl@{}}
\toprule
Result & Statement & Status \\
\midrule
Quality evaluation & Bag: PTIME. Set: PTIME for hierarchical, \#P-hard for non-hierarchical with all atoms probabilistic & \Cref{thm:hardness} \\
Plan invariance & Under PD policies, CASQO separates into classical cost-based QO plus a plan-independent policy choice & \Cref{prop:invariance} \\
PC unsoundness & Selection pushdown is not quality-sound for population-calibrated policies & \Cref{prop:pc-unsound} \\
Per-tuple decomposition & Two prices push an output target down to independent per-tuple decisions & \Cref{prop:lagrange} \\
Instance-level drop & NP-hard, even with one filter, one join, and $p\equiv1$ & \Cref{prop:np} \\
Confidence-centric skipping & Unprocessed instances can be skipped unscored once a bound on their contribution meets the target & \Cref{prop:skip} \\
Greedy approximation & Ranking by saved cost$/c(i)$ gives a knapsack-style guarantee under known $p_i$ and fan-out; coverage structure when witnesses overlap & open \\
Operator-level recall & Under independence, minimizing cost s.t.\ $\sum_j\log R_j\ge\log\rho$ is convex when per-operator cost--recall curves are convex & \Cref{prop:convex} \\
Precision coupling & No order-independent threshold assignment exists in general & open \\
\bottomrule
\end{tabularx}
\end{table}

\paragraph{Tractable class.}
Queries with acyclic relational parts, for which fan-out can be computed with counting semijoins in linear time, are tractable under bag semantics, and under set semantics if the semantic predicates are hierarchical. This matches the Yannakakis-style reductions already used in practice~\citep{jevdb2026}.

\section{Simulations}\label{sec:experiments}

\newcommand{\simBandPush}{$[0.16,0.77)$}
\newcommand{\simBandPull}{$[0.60,0.88)$}
\newcommand{\simIdentical}{all}

\newcommand{\simSaving}{81\%}

\newcommand{\psCostSpread}{1,261}
\newcommand{\psQmin}{0.00}
\newcommand{\psAdmSpread}{307}
\newcommand{\psCostOnlyCost}{2,621}
\newcommand{\psCostOnlyRec}{0.67}
\newcommand{\psCostOnlyPrec}{0.89}
\newcommand{\psEscAllCost}{102,817}
\newcommand{\psEscAllRec}{1.00}
\newcommand{\psEscAllPrec}{1.00}
\newcommand{\psOpTargetCost}{47,461}
\newcommand{\psOpTargetRec}{0.91}
\newcommand{\psOpTargetPrec}{0.98}
\newcommand{\psBestBandCost}{37,690}
\newcommand{\psBestBandRec}{0.91}
\newcommand{\psBestBandPrec}{0.94}
\newcommand{\psOursCost}{6,492}
\newcommand{\psOursRec}{0.90}
\newcommand{\psOursPrec}{0.93}
\newcommand{\swRatioNinety}{5.2}
\newcommand{\swRatioMin}{1.9}
\newcommand{\swRatioMax}{5.3}

\newcommand{\earlyUnscored}{60\%}
\newcommand{\earlySaveOne}{7\%}
\newcommand{\earlySaveTen}{34\%}
\newcommand{\earlyIndep}{17\%}

\newcommand{\earlyStop}{653}
\newcommand{\earlyStopZone}{260}
\newcommand{\earlyHeadShare}{98\%}

\newcommand{\ordDesc}{58\%}
\newcommand{\ordRand}{1\%}
\newcommand{\ordPrior}{24\%}
\newcommand{\zoneRand}{58\%}
\newcommand{\zoneClu}{82\%}
\newcommand{\zoneCluSave}{12\%}
\newcommand{\zoneCluRec}{0.896}

We illustrate the definitions with simulations of the running example. All parameters are stated below; experiments on a real engine, e.g., JEVDB with Jev and an LLM judge on SemBench~\citep{lao2026sembench}, are left for future work. The claims we test concern how quality and cost depend on confidences and fan-out, not how well a particular model is calibrated, and for these claims simulation is sufficient. The workloads are synthetic and were designed to make each effect visible. The reported magnitudes therefore illustrate what the definitions predict and should not be read as estimates of the gains of a deployed system. Two design choices are important. In E2, reviews of popular products are uncertain but mostly positive. In an earlier configuration, in which these reviews had confidences around $0.5$, pushdown escalated them anyway, and both plans met the target about equally often ($93\%$ and $96\%$). In E3 and in the plan-space figures, confidence is drawn independently of fan-out. If the uncertain reviews are exactly those with high fan-out, a single band already escalates them, and per-tuple prices save only about $1\%$. All numbers are produced by \texttt{python3 sim/simulate.py} (standard library only), which generates the tables and plot data used in this section.

\paragraph{Workload.}
We use the query of the running example with $2{,}000$ reviews per instance. A review belongs to a popular product with probability $0.1$ (fan-out uniform in $[50,400]$) and otherwise to a niche product (fan-out uniform in $[1,8]$). We first draw the confidence $p$ of the decision model and then the oracle label $\sstar\sim\mathrm{Bernoulli}(p)$, which makes the model calibrated. Reviews of niche products are mostly unambiguous, with $p\sim\mathrm{Beta}(0.3,8)$ with probability $0.85$ and $p\sim\mathrm{Beta}(30,1)$ otherwise. Reviews of popular products are likely positive but less certain, with $p\sim\mathrm{Beta}(8,2)$ (mean $0.8$); this creates the covariate shift of \cref{prop:pc-unsound}. A decision-model call costs $1$ unit, an LLM call (escalation, always answered correctly) costs $50$, and downstream work costs $0.05$ per joined row. The target is an output recall and precision of at least $0.9$. Each experiment averages over $200$ random instances ($50$ for E3), and $\pm$ denotes one standard deviation.

\begin{table}[t]
\centering
\caption{E1 and E2: pushdown vs.\ pull-up under a pointwise-deterministic band $[0.1,0.9)$ and under bands calibrated on a sample of $200$ instances of each plan's own population (operator-level target $0.95$). Quality is measured on output rows. Cost is relative to PD pushdown; pull-up evaluates every joined row and does not cache decisions. Admissible: fraction of instances meeting the output target $(0.9,0.9)$.}\label{tab:e12}
\small
\begin{tabular}{@{}lccccr@{}}
\toprule
Plan & Output recall & Output precision & F1 & Cost & Admissible \\
\midrule
PD, pushdown & 0.997\,{\scriptsize$\pm$0.001} & 0.983\,{\scriptsize$\pm$0.011} & 0.990\,{\scriptsize$\pm$0.006} & 1.00 & 100\% \\
PD, pull-up & 0.997\,{\scriptsize$\pm$0.001} & 0.983\,{\scriptsize$\pm$0.011} & 0.990\,{\scriptsize$\pm$0.006} & 88.22 & 100\% \\
PC, pushdown & 0.992\,{\scriptsize$\pm$0.020} & 0.924\,{\scriptsize$\pm$0.064} & 0.956\,{\scriptsize$\pm$0.037} & 0.86 & 62\% \\
PC, pull-up & 0.950\,{\scriptsize$\pm$0.020} & 0.977\,{\scriptsize$\pm$0.021} & 0.963\,{\scriptsize$\pm$0.016} & 71.40 & 96\% \\
\bottomrule
\end{tabular}

\end{table}

\paragraph{How much does plan choice matter?}
\Cref{fig:planspace} shows the plan space for one simulated database, with confidence independent of fan-out as in E3. Plans for the same query differ in cost by a factor of up to $\psCostSpread$ and cover the whole range of output quality, so an optimizer can make costly mistakes in both dimensions. Moreover, the two dimensions are not aligned; among the plans that meet the target, cost still varies by a factor of $\psAdmSpread$. The plans chosen by different strategies are far apart:
\begin{itemize}[leftmargin=*,nosep]
  \item a \emph{cost-only} optimizer with a fixed threshold selects the cheapest plan (cost \psCostOnlyCost), which returns only a fraction \psCostOnlyRec{} of the true orders;
  \item \emph{escalating all decisions} gives the exact result at cost \psEscAllCost;
  \item an \emph{operator-level target} of $0.95$ with JEVDB-style calibration meets the output target at cost \psOpTargetCost, and the best single band chosen by predicted output quality at cost \psBestBandCost;
  \item \emph{per-tuple prices} (\cref{prop:lagrange}) meet the target at cost \psOursCost.
\end{itemize}
\Cref{fig:sweep} repeats the comparison for targets between $0.8$ and $0.98$. Per-tuple pricing is cheaper than the best operator-level band by a factor of $\swRatioMin$ to $\swRatioMax$.

\begin{figure}[t]
\centering
\begin{tikzpicture}
\begin{axis}[paper, width=0.95\linewidth, height=7.6cm,
  xmode=log, xmin=1500, xmax=4e6, ymin=0, ymax=1.06,
  xlabel={realized cost (log scale)},
  ylabel={realized $\min(\Rec,\Prec)$ on output rows},
  legend columns=3, legend style={at={(0.5,1.02)},anchor=south,/tikz/every even column/.append style={column sep=8pt}}]
  \addplot[draw=none,fill=cAdm,forget plot] coordinates {(1500,0.9) (4e6,0.9) (4e6,1.06) (1500,1.06)} \closedcycle;
  \draw[cAqua!70!black,dashed] (axis cs:1500,0.9) -- (axis cs:4e6,0.9);
  \node[font=\footnotesize,color=cAqua!50!black,anchor=north east] at (axis cs:3.8e6,1.055) {output target $(0.9,0.9)$ met};
  \addplot[only marks, mark=*, mark size=1.2pt, color=cBlue, fill opacity=0.25, draw opacity=0.35]
    table[x=cost, y=q, restrict expr to domain={\thisrow{family}}{0:0}] {sim/data/plan_space.dat};
  \addplot[only marks, mark=square*, mark size=1.1pt, color=cOrange, fill opacity=0.25, draw opacity=0.35]
    table[x=cost, y=q, restrict expr to domain={\thisrow{family}}{1:1}] {sim/data/plan_space.dat};
  \addplot[only marks, mark=triangle*, mark size=1.6pt, color=cAqua!80!black, fill opacity=0.5, draw opacity=0.7]
    table[x=cost, y=q, restrict expr to domain={\thisrow{family}}{2:2}] {sim/data/plan_space.dat};
  \legend{pushdown + band, pull-up + band, pushdown + per-tuple prices}
  \node[circle,draw=black,fill=white,inner sep=0.6pt,font=\scriptsize\bfseries] at (axis cs:2621,0.6749) {1};
  \node[circle,draw=black,fill=white,inner sep=0.6pt,font=\scriptsize\bfseries] at (axis cs:102817,1.0) {2};
  \node[circle,draw=black,fill=white,inner sep=0.6pt,font=\scriptsize\bfseries] at (axis cs:47461,0.9142) {3};
  \node[circle,draw=black,fill=white,inner sep=0.6pt,font=\scriptsize\bfseries] at (axis cs:37690,0.9096) {4};
  \node[circle,draw=black,fill=white,inner sep=0.6pt,font=\scriptsize\bfseries] at (axis cs:6492,0.8986) {5};
\end{axis}
  \node[draw=cInk,fill=white,align=left,font=\footnotesize,anchor=north,inner sep=4pt] at ([yshift=-4pt]current axis.below south west -| current axis.south) {%
    \textbf{1} cost-only (fixed $0.5$ threshold): cost \psCostOnlyCost, recall \psCostOnlyRec\\
    \textbf{2} escalate all: cost \psEscAllCost\\
    \textbf{3} operator-level target $0.95$: cost \psOpTargetCost\\
    \textbf{4} best single band (predicted output quality): \psBestBandCost\\
    \textbf{5} \textbf{ours, per-tuple prices: cost \psOursCost}};
\end{tikzpicture}
\caption{Plan space of one simulated database ($2{,}000$ reviews, confidence independent of fan-out). Each point is a candidate plan: every band on a $0.05$ grid under pushdown and pull-up, and the per-tuple priced policies of \cref{prop:lagrange} under pushdown. Plans differ by up to $\psCostSpread\times$ in cost and from $\psQmin$ to $1$ in output quality; even among plans that meet the target, cost varies $\psAdmSpread\times$. The cost-only plan is cheap but misses about a third of the true orders. Our policy meets the predicted target (realized recall \psOursRec, precision \psOursPrec) at about $1/16$ of the cost of escalating all decisions and $1/7$ of the cost of an operator-level target.}\label{fig:planspace}
\end{figure}
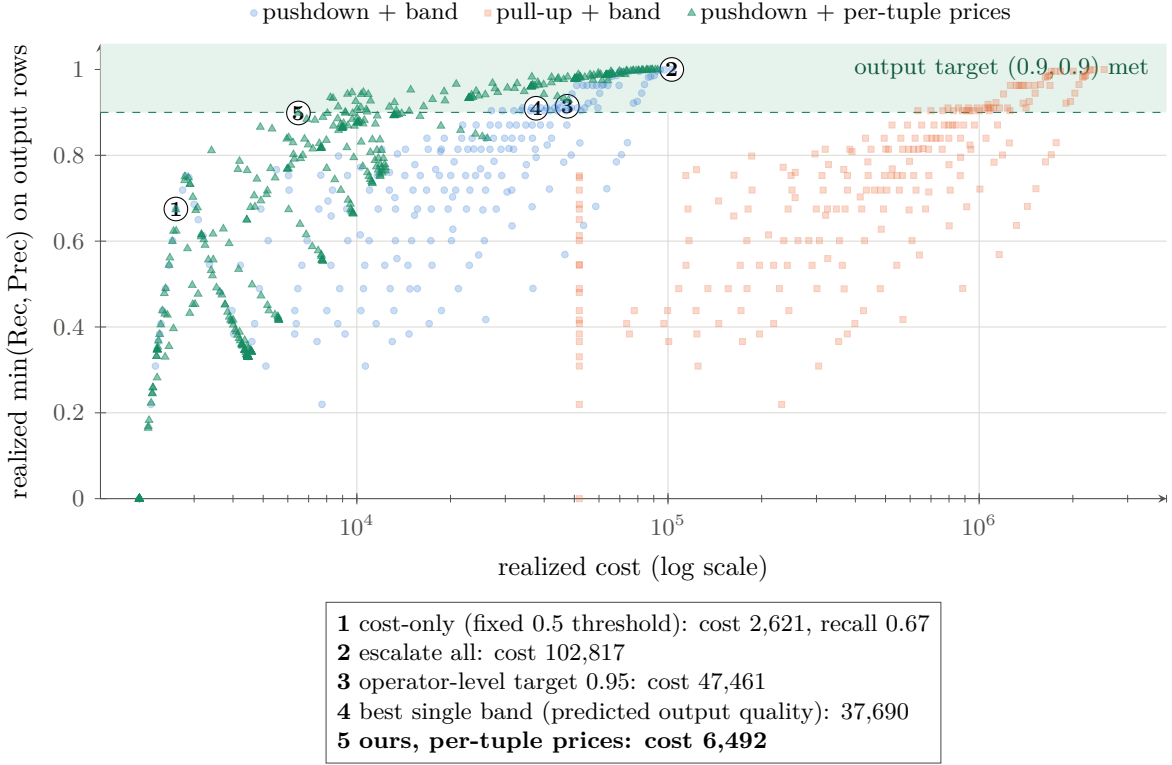
\begin{figure}[t]
\centering
\begin{tikzpicture}
\begin{axis}[paper, width=0.78\linewidth, height=6cm, clip=false,
  xmin=0.795, xmax=1.0, ymin=0, ymax=112,
  xtick={0.80,0.84,0.88,0.92,0.96}, xticklabel style={/pgf/number format/fixed,/pgf/number format/precision=2},
  xlabel={output target $\rho=\gamma$}, ylabel={cost of the cheapest plan meeting it (thousands)}]
  \addplot[cInk, dashed, thick] table[x=target, y expr=\thisrow{esc}/1000] {sim/data/sweep.dat};
  \addplot[cBlue, very thick, mark=*, mark size=1.8pt] table[x=target, y expr=\thisrow{op}/1000] {sim/data/sweep.dat};
  \addplot[cOrange, very thick, mark=square*, mark size=1.8pt] table[x=target, y expr=\thisrow{ins}/1000] {sim/data/sweep.dat};
  \node[font=\footnotesize,color=cInk,anchor=south west] at (axis cs:0.80,102.82) {escalate everything};
  \node[font=\footnotesize,anchor=west,align=left] at (axis cs:0.982,70.276) {operator-level\\band};
  \node[font=\footnotesize,anchor=west,align=left] at (axis cs:0.982,36.708) {per-tuple\\prices (ours)};
  \draw[{Stealth[length=2mm]}-{Stealth[length=2mm]},cInk] (axis cs:0.90,7.128) -- (axis cs:0.90,37.146)
    node[midway,right,font=\footnotesize,align=left,fill=white,inner sep=1pt] {$\swRatioNinety\times$ cheaper\\at $0.9$};
\end{axis}
\end{tikzpicture}
\caption{Cost of meeting an output-level target, averaged over $10$ simulated databases (confidence independent of fan-out). Both policies are chosen on predicted quality and meet the target on realized labels on average. Per-tuple pricing is cheaper than the best operator-level band by a factor of $\swRatioMin$ to $\swRatioMax$. The gap narrows as $\rho$ approaches $1$, since then almost all uncertain tuples have to be escalated.}\label{fig:sweep}
\end{figure}
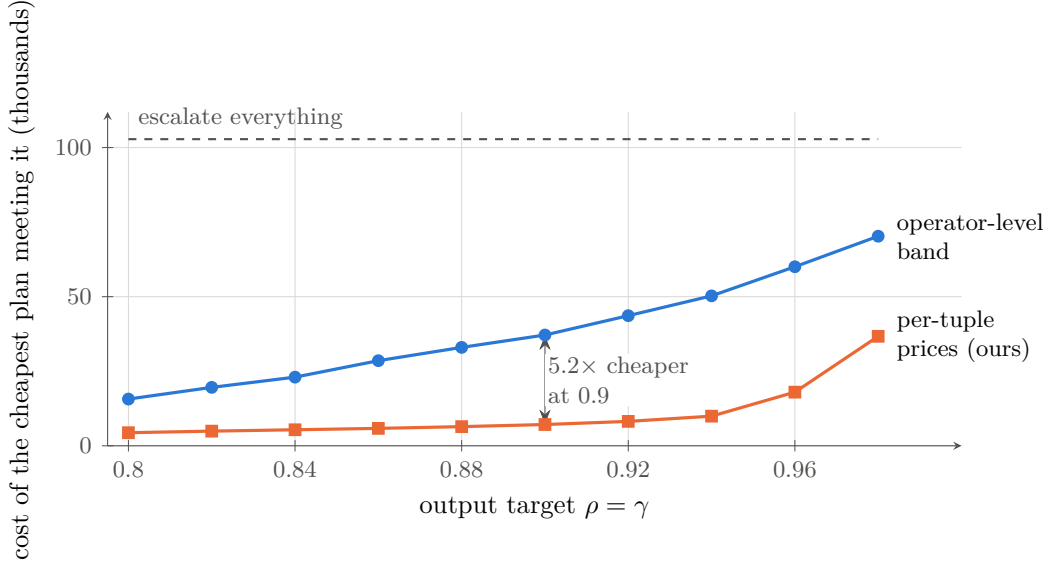

\paragraph{E1: plan invariance under a PD policy.}
With a fixed band, pushdown and pull-up return identical results on \simIdentical{} $200$ instances (\cref{tab:e12}, first two rows), as predicted by \cref{prop:invariance}, and differ only in cost. For PD policies, the choice between the two plans is thus a classical cost-based decision.

\paragraph{E2: plan dependence under a calibrated band.}
If each plan fits its band on a sample of the instances it evaluates, the resulting bands differ, with \simBandPush{} for pushdown and \simBandPull{} for pull-up on average. Pushdown samples reviews, among which popular products are rare, and therefore accepts popular reviews with $p\approx0.8$. Each such review contributes $50$ to $400$ rows that are wrong with probability $0.2$. As a result, output precision drops, and pushdown meets the target on far fewer instances than pull-up (\cref{tab:e12}, last two rows). This is the effect of \cref{prop:pc-unsound}, here with a random calibration sample; \cref{fig:e2e4}(a) shows the variation across databases. Pull-up meets the target but is about $70$ times more expensive, since it evaluates every joined row. A plan that is as cheap as pushdown and as accurate as pull-up must choose its decisions based on output quality rather than operator quality; E3 evaluates this approach.

\begin{table}[t]
\centering
\caption{E3: cheapest operator-level band vs.\ instance-level policy (\cref{prop:lagrange}, grid search over the two prices), both chosen on \emph{predicted} output quality to meet $(0.9,0.9)$ and evaluated on realized quality. Workload with confidence independent of fan-out; $50$ instances.}\label{tab:e3}
\small
\begin{tabular}{@{}lrrcc@{}}
\toprule
Policy & Cost & Escalations & Output recall & Output precision \\
\midrule
Operator-level band & 37,256 & 688 & 0.917\,{\scriptsize$\pm$0.027} & 0.911\,{\scriptsize$\pm$0.033} \\
Instance-level (priced) & 7,164 & 87 & 0.921\,{\scriptsize$\pm$0.022} & 0.910\,{\scriptsize$\pm$0.021} \\
\bottomrule
\end{tabular}

\end{table}

\paragraph{E3: pushing the target down to tuples.}
To isolate the effect of fan-out, we draw confidence independently of fan-out in this experiment ($p\sim\mathrm{Beta}(0.7,3)$ with probability $0.8$, else $\mathrm{Beta}(3,0.7)$). The operator-level policy considers all bands $[\ell,h)$ on a grid with step $0.02$, and the instance-level policy assigns actions with the per-tuple rule of \cref{prop:lagrange}. Both select the cheapest policy whose \emph{predicted} output quality meets the target, using only confidences and fan-outs, and both meet the target on the realized labels (\cref{tab:e3}). The instance-level policy is \simSaving{} cheaper. It escalates uncertain reviews with many rows and leaves uncertain reviews with few rows to the cheap model, whereas a band has to escalate every review in its range.

\begin{table}[t]
\centering
\caption{E4: predicted vs.\ realized output quality of the PD band policy, averaged over $200$ instances. Overconfident: logit of $p$ scaled by $2.5$. PPI: contribution-weighted correction with $100$ gold labels drawn in proportion to fan-out. Last column: mean absolute error of predicted precision per instance.}\label{tab:e4}
\small
\begin{tabular}{@{}lccccc@{}}
\toprule
Decision model & Pred.\ recall & True recall & Pred.\ precision & True precision & $|$Error$|$ (precision) \\
\midrule
Calibrated & 0.997 & 0.997 & 0.983 & 0.982 & 0.008\,{\scriptsize$\pm$0.006} \\
Overconfident & 1.000 & 0.994 & 0.981 & 0.872 & 0.109\,{\scriptsize$\pm$0.027} \\
Overconfident + PPI & 0.995 & 0.994 & 0.874 & 0.872 & 0.028\,{\scriptsize$\pm$0.020} \\
\bottomrule
\end{tabular}

\end{table}

\begin{figure}[t]
\centering
\begin{tikzpicture}
\begin{groupplot}[group style={group size=2 by 1, horizontal sep=1.7cm}, paper, width=0.47\linewidth, height=6cm]
\nextgroupplot[title={\small (a) E2: output precision per database},
  xmin=0.4, xmax=2.6, xtick={1,2}, xticklabels={PC pushdown, PC pull-up}, xmajorgrids=false,
  ymin=0.65, ymax=1.01, ylabel={realized output precision}]
  \addplot[draw=none,fill=cAdm,forget plot] coordinates {(0.4,0.9) (2.6,0.9) (2.6,1.01) (0.4,1.01)} \closedcycle;
  \draw[cAqua!70!black,dashed] (axis cs:0.4,0.9) -- (axis cs:2.6,0.9) node[pos=0.02,above right,font=\scriptsize,color=cAqua!50!black,inner sep=1pt]{target};
  \addplot[only marks, mark=*, mark size=1.2pt, color=cBlue, fill opacity=0.5, draw opacity=0.6]
    table[x expr={1+0.55*\thisrow{jitter}}, y=push] {sim/data/e2_strip.dat};
  \addplot[only marks, mark=square*, mark size=1.1pt, color=cOrange, fill opacity=0.5, draw opacity=0.6]
    table[x expr={2+0.55*\thisrow{jitter}}, y=pull] {sim/data/e2_strip.dat};
\nextgroupplot[title={\small (b) E4: predicted vs.\ realized precision},
  xmin=0.75, xmax=1.0, ymin=0.75, ymax=1.0, xlabel={predicted output precision}, ylabel={realized output precision},
  xticklabel style={/pgf/number format/fixed,/pgf/number format/precision=2},
  yticklabel style={/pgf/number format/fixed,/pgf/number format/precision=2},
  legend style={at={(0.02,0.98)},anchor=north west,fill=white,fill opacity=0.9,text opacity=1,draw=cGrid}]
  \addplot[cInk,thin,forget plot] coordinates {(0.75,0.75) (1,1)};
  \addplot[only marks, mark=*, mark size=1.2pt, color=cBlue, fill opacity=0.5, draw opacity=0.6]
    table[x=calpred, y=caltrue] {sim/data/e4_scatter.dat};
  \addplot[only marks, mark=square*, mark size=1.1pt, color=cOrange, fill opacity=0.5, draw opacity=0.6]
    table[x=overpred, y=overtrue] {sim/data/e4_scatter.dat};
  \addplot[only marks, mark=triangle*, mark size=1.5pt, color=cAqua!80!black, fill opacity=0.5, draw opacity=0.7]
    table[x=ppipred, y=ppitrue] {sim/data/e4_scatter.dat};
  \legend{calibrated, overconfident, overconfident + PPI}
\end{groupplot}
\end{tikzpicture}
\caption{(a) Each dot is one of $200$ simulated databases. A band calibrated on the pushdown population often falls below the output precision target; calibrated on the pull-up population it rarely does (\cref{prop:pc-unsound}). (b) Predicted output precision from $\sum_ic_i(1-\mathrm{conf}_i)$ against realized precision. Calibrated confidences lie on the diagonal; an overconfident model predicts about $0.98$ for databases that deliver $0.87$; contribution-weighted PPI with $100$ labels moves the predictions back to the diagonal, with more spread.}\label{fig:e2e4}
\end{figure}
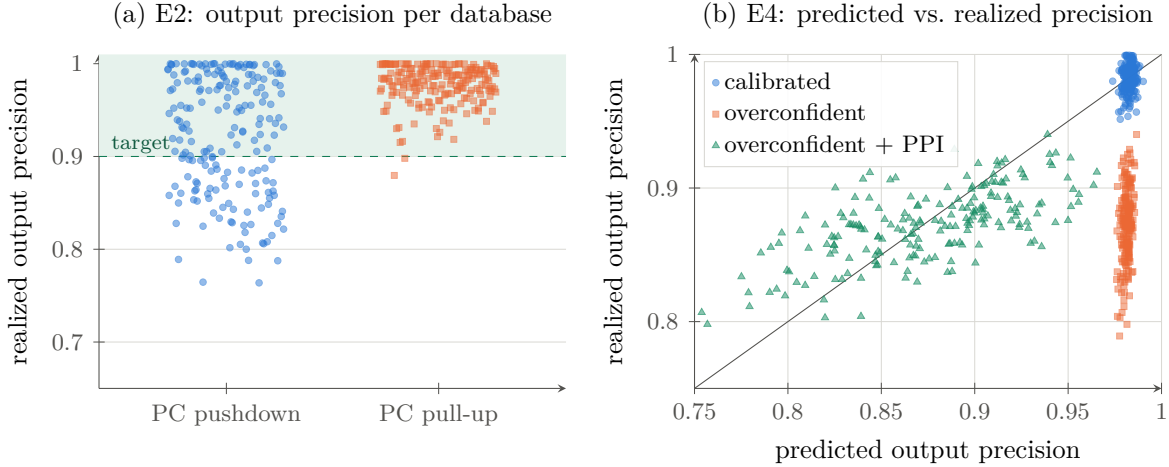

\paragraph{E4: accuracy of the prediction.}
For a calibrated model, the output quality predicted from $\sum_ic_i(1-\mathrm{conf}_i)$ is within about $0.01$ of the realized quality on each instance (\cref{tab:e4} and \cref{fig:e2e4}(b)). An overconfident model, whose confidences are pushed toward $0$ and $1$, predicts a precision of $0.98$, while the realized precision is $0.87$; the hidden errors are concentrated in reviews with high fan-out. Contribution-weighted PPI with $100$ labels removes this bias on average and reduces the per-instance error from about $0.11$ to about $0.03$. In summary, the predicted quality used for equivalence and optimization is accurate for calibrated models and can be corrected with a small gold sample otherwise.

\begin{table}[t]
\centering
\caption{E5: scoring every review vs.\ staged processing with confidence-centric skipping (\cref{prop:skip}), both chosen on predicted output quality to meet $(0.9,0.9)$. Heavy-tailed fan-out, $2{,}000$ reviews, $20$ instances. Model cost: cost of one decision-model call (one LLM call costs $50$).}\label{tab:e5}
\small
\begin{tabular}{@{}rlrrrcc@{}}
\toprule
Model cost & Policy & Reviews scored & Escalations & Cost & Output recall & Output precision \\
\midrule
1 & score all & 2,000 & 145 & 18,389 & 0.924\,{\scriptsize$\pm$0.021} & 0.918\,{\scriptsize$\pm$0.016} \\
1 & staged + skip & 806 & 144 & 17,063 & 0.915\,{\scriptsize$\pm$0.019} & 0.916\,{\scriptsize$\pm$0.016} \\
10 & score all & 2,000 & 145 & 36,389 & 0.924\,{\scriptsize$\pm$0.021} & 0.918\,{\scriptsize$\pm$0.016} \\
10 & staged + skip & 712 & 157 & 24,149 & 0.925\,{\scriptsize$\pm$0.022} & 0.913\,{\scriptsize$\pm$0.014} \\
\bottomrule
\end{tabular}

\end{table}

\begin{figure}[t]
\centering
\begin{tikzpicture}
\begin{axis}[paper, width=0.8\linewidth, height=6cm, clip=false,
  xmin=0, xmax=2000, ymin=0, ymax=1.02,
  xlabel={reviews scored},
  ylabel={recall lower bound (\cref{prop:skip})},
  legend columns=3, legend style={at={(0.5,-0.25)},anchor=north,font=\scriptsize,/tikz/every even column/.append style={column sep=6pt}}]
  \draw[cAqua!70!black,dashed] (axis cs:0,0.9) -- (axis cs:2000,0.9) node[pos=0.6,below,font=\scriptsize,color=cAqua!50!black,inner sep=2pt,fill=white]{target $\rho=0.9$};
  \addplot[cAqua!80!black, very thick] table[x=k, y=zone] {sim/data/early_trace.dat};
  \addplot[cBlue, very thick] table[x=k, y=desc] {sim/data/early_trace.dat};
  \addplot[cOrange, very thick, dashed] table[x=k, y=rand] {sim/data/early_trace.dat};
  \legend{decreasing fan-out + clustered zone map, decreasing fan-out, storage (random) order}
  \draw[cAqua!50!black,thick] (axis cs:\earlyStopZone,0) -- (axis cs:\earlyStopZone,0.9);
  \draw[cBlue!70!black,thick] (axis cs:\earlyStop,0) -- (axis cs:\earlyStop,0.9);
  \node[font=\footnotesize,anchor=south,align=center,fill=white,inner sep=1pt] at (axis cs:\earlyStopZone,0.42) {stop at\\\earlyStopZone};
  \node[font=\footnotesize,anchor=south,align=center,fill=white,inner sep=1pt] at (axis cs:\earlyStop,0.3) {stop at\\\earlyStop};
\end{axis}
\end{tikzpicture}
\caption{Confidence-centric skipping on one simulated database with heavy-tailed fan-out (same prices for all three runs). The lower bound treats every unscored row as a true positive, or uses a block's maximum confidence when a zone-map summary is available; it requires no scoring of the unscored reviews. Once it reaches the target, the remaining reviews are skipped. In decreasing order of fan-out, the first \earlyStop{} reviews carry \earlyHeadShare{} of the output rows; with a zone map clustered by confidence, the bound reaches the target after \earlyStopZone{} reviews. In storage order the bound never reaches the target.}\label{fig:early}
\end{figure}
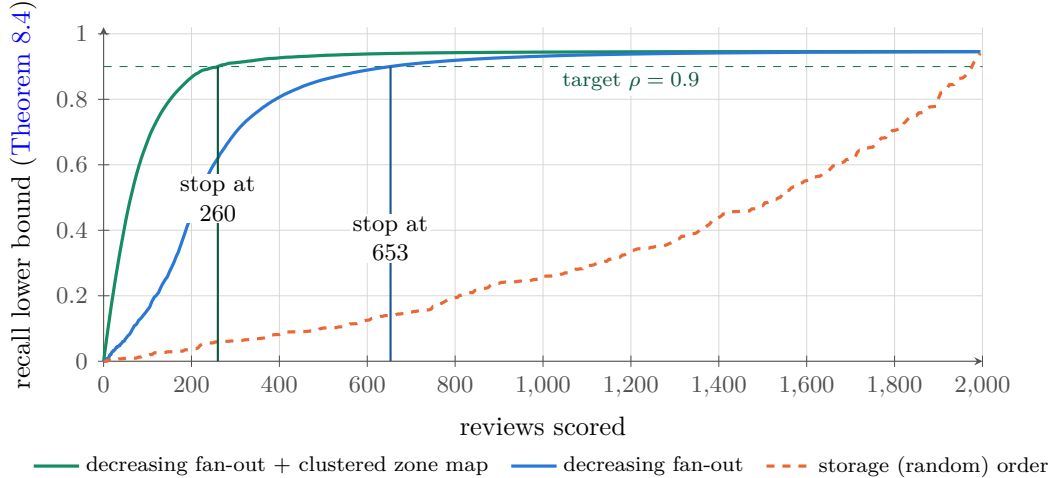

\paragraph{E5: confidence-centric skipping.}
This experiment uses a workload with heavy-tailed fan-out ($f=\lfloor U^{-1/0.3}\rfloor$ for uniform $U$, capped at $2{,}000$), with confidence drawn independently of fan-out as in E3. We compare the cheapest per-tuple priced policy that scores every review with the cheapest staged policy of \cref{prop:skip}, which processes reviews in decreasing order of fan-out and stops once the worst-case bound meets the target. Both meet the target on realized labels (\cref{tab:e5}). The staged policy leaves \earlyUnscored{} of the reviews unscored. With a decision-model call costing $1$, most of the cost comes from escalations and downstream rows, and the total cost drops by \earlySaveOne{}; with a call costing $10$, e.g., for long documents or a larger model, it drops by \earlySaveTen{}. In the run of \cref{fig:early}, the first \earlyStop{} reviews carry \earlyHeadShare{} of the output rows, and after them the remaining reviews cannot violate the target. With the less skewed fan-out of E3, \earlyIndep{} of the reviews are skipped. Hence, the assumption that every instance is scored can be relaxed in practice, and the benefit grows with the skew of fan-out and the cost of the decision model.

\begin{table}[t]
\centering
\caption{E6: processing order and confidence summaries for confidence-centric skipping, on the E5 workload ($10$ instances). Prices are re-tuned for each configuration. Unscored: fraction of reviews never scored; cost saving relative to the cheapest score-all policy; both given as mean / max over instances. Zone maps store the maximum confidence of each block of $64$ reviews.}\label{tab:e6}
\small
\begin{tabular}{@{}lcccc@{}}
\toprule
Order and summary & Unscored & Cost saving & Output recall & Output precision \\
\midrule
Decreasing fan-out & 58\% / 67\% & 7.0\% / 9.3\% & 0.912\,{\scriptsize$\pm$0.017} & 0.913\,{\scriptsize$\pm$0.016} \\
Storage (random) order & 1\% / 2\% & 0.7\% / 1.7\% & 0.914\,{\scriptsize$\pm$0.021} & 0.916\,{\scriptsize$\pm$0.016} \\
Increasing fan-out & 0\% / 0\% & 0.7\% / 1.7\% & 0.911\,{\scriptsize$\pm$0.015} & 0.916\,{\scriptsize$\pm$0.016} \\
Fan-out $\times$ noisy prior & 24\% / 29\% & 2.6\% / 3.3\% & 0.920\,{\scriptsize$\pm$0.018} & 0.917\,{\scriptsize$\pm$0.016} \\
Zone map, random layout & 58\% / 67\% & 7.1\% / 9.3\% & 0.912\,{\scriptsize$\pm$0.017} & 0.914\,{\scriptsize$\pm$0.016} \\
Zone map, clustered & 82\% / 87\% & 11.8\% / 13.8\% & 0.896\,{\scriptsize$\pm$0.015} & 0.912\,{\scriptsize$\pm$0.016} \\
\bottomrule
\end{tabular}

\end{table}

\paragraph{E6: order and summaries.}
\Cref{tab:e6} varies the processing order and the bounds $u_i$ of \cref{prop:skip}, re-tuning the prices for each configuration. Decreasing fan-out leaves \ordDesc{} of the reviews unscored; storage order and increasing fan-out leave almost none (\ordRand{} and less), because the residual mass shrinks too slowly. Ordering by fan-out times a noisy prior on the confidence (logit noise with standard deviation $1.5$) leaves only \ordPrior{} unscored, as the prior moves high-fan-out reviews with low prior confidence to the back, where their rows keep the worst-case bound low. A zone map whose blocks group reviews of similar confidence leaves \zoneClu{} unscored and saves \zoneCluSave{} of the cost; with blocks in random layout, every block maximum is close to $1$ and the zone map gives no benefit (\zoneRand{}). With the tighter bound, the realized output recall is closer to the target (\zoneCluRec{} on average, with predicted recall at least $0.9$). \Cref{fig:early} shows the bound for three of these configurations on a single database.

\paragraph{Limitations.}
The simulations assume independent decisions and a perfect judge, and, except in E4, a calibrated model. Correlated errors (e.g., one review in many pairs), judge errors on the most difficult instances (\cref{ex:judge}), and the calibration of real models across the populations induced by different plans have to be studied on a real engine.

\section{Related Work and Positioning}\label{sec:related}

The closest prior work allocates accuracy across ML predicates or guarantees per-operator quality; none defines output-level targets with join fan-out coupling or a semantic plan equivalence.

\paragraph{Expensive predicates in classical optimization.}
Predicate migration~\citep{hellerstein1993migration,hellerstein1998expensive} and the optimization of user-defined predicates~\citep{chaudhuri1999udp} place costly predicates in a plan by trading their cost against their selectivity, and adaptive schemes such as eddies~\citep{avnur2000eddies} and adaptive filter ordering~\citep{babu2004filters} reorder them at run time. These predicates are exact, so every placement returns the same answer and only cost is at stake. This corresponds to the setting of pointwise-deterministic policies (\cref{prop:invariance}); we study the case in which predicate placement also affects the result. Sampling-based approximate query processing such as BlinkDB~\citep{agarwal2013blinkdb} bounds the error of aggregates over samples of the data rather than the error of model decisions. Data skipping with min/max summaries such as small materialized aggregates~\citep{moerkotte1998sma} avoids reading blocks whose value range cannot satisfy a predicate and leaves the result unchanged; our confidence-centric skipping (\cref{sec:skip}) instead stops evaluating instances once the output target is guaranteed, which may change the result. Its stopping rule resembles that of the threshold algorithm for top-$k$ aggregation~\citep{fagin2003ta}, which stops once a bound on the unseen objects certifies the answer.

\paragraph{Accuracy targets for ML predicates and proxies.}
Probabilistic predicates~\citep{lu2018pp} split a query-level accuracy target across cheap filters placed ahead of expensive UDFs and choose among them by cost; they handle filters only, with operator-level budgets and no join fan-out. NoScope~\citep{kang2017noscope}, Tahoma~\citep{anderson2019tahoma}, and CORE~\citep{yang2022core} build cascades or proxy models for ML predicates over video and images, and BlazeIt~\citep{kang2019blazeit} and ABAE~\citep{kang2021abae} use proxies to answer aggregation and limit queries with accuracy guarantees. SUPG~\citep{kang2020supg} gives precision/recall-target selection with a proxy and statistical guarantees for a single selection, and TASTI~\citep{kang2022tasti} builds reusable semantic indexes that provide proxy scores across queries. These systems define accuracy for a single operator or a specific query shape and do not consider fan-out or the calibration population in plans with multiple operators.

\paragraph{Semantic operator systems.}
LOTUS~\citep{patel2025lotus} provides per-operator accuracy guarantees using model cascades; FrugalGPT~\citep{chen2023frugalgpt} applies cascades to individual LLM prompts. Palimpzest and its optimizer Abacus~\citep{liu2025palimpzest,russo2026abacus} search a cost--quality Pareto frontier over physical implementations, with quality estimated empirically per operator. DocETL~\citep{shankar2025docetl} rewrites document-processing pipelines with an LLM agent and evaluates candidate rewrites with LLM-based validation, and CAESURA~\citep{urban2023caesura} uses a language model to plan multi-modal queries; both assess plan quality empirically, without a model of how decision errors propagate. ThalamusDB~\citep{jo2024thalamusdb} supports approximate processing with error bounds and user labeling. Featurized-decomposition joins~\citep{zeighami2025fdj} prune semantic joins soundly and belong to our class of sound pruning. JEVDB~\citep{jevdb2026} places calibrated decision models on the fast path with semantic Bloom filters and per-predicate escalation bands calibrated against an LLM judge, restores predicate pushdown as the default on cost grounds, and leaves query optimization to future work. Larch~\citep{zhao2026larch} and iPDB~\citep{kumarasinghe2026ipdb} optimize semantic predicates for cost: Larch learns a short-circuit evaluation order over a boolean tree of semantic filters to minimize LLM tokens, and iPDB pulls predict operators up the plan and reorders semantic selects inside DuckDB. Neither system defines an output-level accuracy target. Larch assumes that reordering does not change output quality, i.e., it treats each predicate's decision as fixed. iPDB targets latency and cost; its join decomposition trades quality for speed under user control, and quality is evaluated empirically (F1).

\paragraph{Foundations we build on.}
Our results build on several lines of work. We use the dichotomy for probabilistic databases~\citep{dalvi2007dichotomy,suciu2011probdb} for quality estimation and provenance semirings~\citep{green2007provenance} for decision lineage. Attribution of output quality to decisions relates to sensitivity analysis~\citep{kanagal2011sensitivity} and tuple Shapley values~\citep{livshits2020shapley}, contributions are computed by message passing as in FAQ~\citep{abokhamis2016faq}, and our oracle-relative estimator is based on prediction-powered inference~\citep{angelopoulos2023ppi}.

\section{Open Problems}\label{sec:open}

Our definitions cover positive SRA with bag outputs. We list several directions in which the scope should be extended.
\begin{itemize}[leftmargin=*]
  \item \textbf{Aggregates and \texttt{GROUP BY}.} A misclassification can move a tuple from one group to another, so errors are not simply false positives or false negatives. Quality should then be defined via relative error or group-level F1. For \texttt{COUNT} and \texttt{SUM}, the cost of dropping a tuple is again proportional to its fan-out.
  \item \textbf{Negation and \texttt{NOT EXISTS}.} A semantic false negative inside a negated subquery becomes a false positive in the output, which reverses the asymmetry between recall and precision discussed in \cref{sec:propagation}.
  \item \textbf{Top-$k$ and ranking.} It is open how ranking quality composes with joins.
  \item \textbf{Multi-query optimization.} Since different queries may have different targets, decisions can only be reused across queries if the confidences $p_i$ are stored rather than Boolean outcomes. A condition index is a special case.
  \item \textbf{Semantic zone maps.} Per-block summaries, such as label sets or the maximum confidence, tighten the bounds of confidence-centric skipping (\cref{sec:skip}) and can also serve for sound or approximate pruning. Our simulations show a benefit only when blocks group instances of similar confidence (\cref{tab:e6}). How to build and maintain such layouts for predicates that are not known in advance, e.g., by clustering on embeddings, is open.
  \item \textbf{Processing order.} Decreasing contribution is a natural order for skipping, but the order also affects partial results, latency, and, under population-calibrated policies, the calibration sample. A joint treatment of order, skipping, and calibration is open.
  \item \textbf{Adaptive re-optimization.} A plan may need to be revised when the observed selectivity or calibration of a semantic predicate differs from the estimates.
\end{itemize}

\section{Conclusion}\label{sec:conclusion}

In semantic query engines, the choice of a plan determines not only the cost of a query but also its result. We argued that optimizing such plans requires a precise definition of what is being optimized, and that calibrated confidences make such a definition possible. A calibrated confidence assigns an expected error to each semantic decision. Weighting these errors by each decision's contribution to the output gives the expected output quality of a plan without labeled data, and the same computation, applied in reverse, assigns a price for output error to every base and intermediate tuple. Based on this, we defined physical plans as logical plans with decision policies, output-level quality targets, a hierarchy of plan equivalence that ends with the set of admissible plans, and the corresponding optimization problem.

Our results delimit the problem. Under pointwise-deterministic policies, accuracy does not depend on the plan, and optimization reduces to the classical case. The interaction between cost and accuracy is caused by calibrated, context-aware, or approximate components. In particular, selection pushdown, the default rewrite in current systems, is not quality-sound once bands are calibrated on the candidates of the plan itself. Expected quality can be computed efficiently under bag semantics and, under set semantics, follows the dichotomy for probabilistic databases when every relation carries a semantic predicate. Selecting tuples to drop is NP-hard, while the output target decomposes into per-tuple decisions with two prices. Moreover, once the decisions taken so far guarantee the target, the remaining tuples can be skipped without being scored, so confidences are only needed for the tuples that matter.

Our simulations illustrate these effects on synthetic data. The next step is to measure them on real engines, models, and benchmarks. Beyond that, the open problems in \cref{sec:open}, in particular aggregates and negation, correlated decisions, and adaptive re-optimization for miscalibrated confidences, remain to be addressed.

\bibliographystyle{plainnat}
\bibliography{refs}
\end{document}